\documentclass[12pt]{article}
\usepackage[font={footnotesize},justification=justified]{caption}
\usepackage{amsmath, mathpazo, amssymb, amsfonts, amsthm}
\usepackage[T1]{fontenc}
\usepackage[colorlinks=true,linkcolor=blue, citecolor=blue, urlcolor=blue]{hyperref}
\usepackage[utf8]{inputenx}
\usepackage{graphicx}
\usepackage{multirow}
\usepackage{dcolumn}
\usepackage{pdflscape}
\usepackage{natbib}
\usepackage{hyperref}
\newcolumntype{d}[1]{D{.}{.}{-1}}
\newcolumntype{s}[1]{D{,}{,}{-1}}
\usepackage[margin=0.9in]{geometry}
\usepackage[normalem]{ulem}
\usepackage{enumitem}
\usepackage{xspace}
\usepackage[group-separator={,},group-minimum-digits=4]{siunitx}
\usepackage[hang,flushmargin]{footmisc}
\usepackage[toc, title]{appendix}
\usepackage{setspace}
\usepackage{bookmark}
\usepackage[svgnames]{xcolor}
\usepackage{booktabs}
\usepackage[colorinlistoftodos,textsize=footnotesize]{todonotes}
\usepackage{xargs}
\usepackage{color-edits}

\addauthor[Brendan]{bl}{blue}
\addauthor[Salvador]{sc}{olive}

\usepackage{fancybox,fancyhdr}
\usepackage{mathtools}
\usepackage{bbm}
\usepackage{graphicx}
\usepackage{caption}
\usepackage{tikz}
\usepackage{subcaption}
\usetikzlibrary{patterns}
\usetikzlibrary{shapes.geometric}
\usetikzlibrary{shapes.misc}
\usetikzlibrary{plotmarks}
\usepackage{standalone}

\usepackage{mathabx}

\theoremstyle{plain}
\newtheorem{definition}{Definition}

\newtheorem{theorem}{Theorem}
\newtheorem*{theorem*}{Theorem}
\newtheorem{lemma}{Lemma}

\newtheorem{claim}{Claim}
\newtheorem{corollary}{Corollary}
\theoremstyle{definition}

\DeclareMathOperator{\E}{\mathbb{E}}

\DeclareMathOperator{\1}{\mathbbm{1}}

\usepackage{array}
\newcommand{\PreserveBackslash}[1]{\let\temp=\\#1\let\\=\temp}
\newcolumntype{C}[1]{>{\PreserveBackslash\centering}p{#1}}
\newcolumntype{R}[1]{>{\PreserveBackslash\raggedleft}p{#1}}
\newcolumntype{L}[1]{>{\PreserveBackslash\raggedright}p{#1}}

\usepackage{clipboard}
\newclipboard{empirical-large-core}

\usepackage[capitalize,noabbrev,nameinlink]{cleveref}

\crefname{assumption}{assumption}{assumptions}
\usepackage{appendix}
\usepackage{etoolbox}
\AtBeginEnvironment{appendices}{\crefalias{section}{appendix}}
\newif\ifinappendix\inappendixfalse
\crefname{apptab}
  {\protect{\ifinappendix\else Appendix \fi}Table}
  {\protect{\ifinappendix\else Appendix \fi}Table}
\AtBeginEnvironment{appendices}{\inappendixtrue \crefalias{table}{apptab}}

\usepackage{titlesec}
\titlespacing*{\section}
{0pt}{2.5ex plus 1ex minus .2ex}{2ex plus .2ex}
\titlespacing*{\subsection}
{0pt}{2.25ex plus 1ex minus .2ex}{1.2ex plus .2ex}
\titlespacing*{\paragraph}
{0pt}{2.25ex plus 1ex minus .2ex}{1em}

\makeatletter 
\renewcommand\paragraph{\@startsection{paragraph}{4}{\z@}%
                                    {1.7ex \@plus1ex \@minus.2ex}%
                                    {-1em}%
                                    {\normalfont\normalsize\bfseries}}
\makeatother

\usepackage{pgfplots}
\pgfplotsset{compat=newest}

\newcommand{\attention}{attention\xspace}

\title{
Subsidizing Sequential Search
}
\author{Salvador Candelas\thanks{Pennsylvania State University. Email: \href{mailto: scandelas@psu.edu}{scandelas@psu.edu}} \and Nicole Immorlica\thanks{Microsoft Research / Yale. Email: \href{mailto: nicimm@gmail.com}{nicimm@gmail.com}} \and Brendan Lucier\thanks{Microsoft Research. Email: \href{mailto: brlucier@microsoft.com}{brlucier@microsoft.com}}
}

\begin{document}

\maketitle

\begin{abstract}

We study markets where firms compete for consumer attention by subsidizing costly product inspection. These subsidies do not change product quality, but they alter the order in which consumers search by lowering inspection costs. We establish a subsidy–sorting principle: in any equilibrium, higher–quality firms provide weakly larger subsidies, leading consumers to search in descending subsidy order. A unique equilibrium survives forward–induction reasoning in the spirit of the Intuitive Criterion: low–quality firms are never inspected, intermediate–quality firms separate with strictly increasing subsidies, and high–quality firms pool at the full subsidy. This equilibrium maximizes information revelation among all possible outcomes and ensures efficient inspection. We then extend the analysis to AI–mediated platforms that can create and price ``inspection tokens.” The platform’s optimal linear pricing leads to excessive inspection relative to the social optimum. While this distortion does not reduce consumer welfare, it reallocates surplus from sellers to the platform and consumers.

\end{abstract}

\vspace{6mm}

\vspace{6mm}

\newline
\onehalfspace

\newpage

\section{Introduction}\label{sec:introduction}

Markets run on attention. Buyers want to find a product that fits their needs, but evaluating options takes time and effort. This challenge is familiar in online markets: a consumer types a query, browses links, and exerts effort before finding what she wants. Firms compete for the consumer's attention not only through price and product features but also by lowering the cost of evaluation, most visibly through sponsored placements and promoted listings. This phenomenon arises also beyond the web. For instance, graduate programs cover travel expenses so prospective students can visit campus and meet faculty; car dealers hand out cash for test drives; galleries and non--profits host receptions to draw people. 
Across these settings, in both online and off-line markets, the consumer faces search frictions and the other side competes to offset these search frictions. 

We formalize this problem as a new kind of ``directed search'': a consumer chooses which options to evaluate, and each evaluation carries a cost. Firms compete for her attention by  subsidizing some or all of her search cost. These subsidies play a dual role. They are instrumental, because they reduce the cost of being considered and enlarge the set of products the consumer is willing to examine.  
They are informative, because only firms confident in high match value find it worthwhile to subsidize. In this way, subsidies both broaden what is searched and determine the order in which search unfolds.

We formalize these ideas in \cref{sec:environment}. We study a sequential–search model with vertical match uncertainty. A consumer seeks a product from one of $n$ firms. Firm $j$ has a privately known type $t_j \in [0,1]$, which denotes the probability that its product is a match for the consumer's needs. Firm types are drawn i.i.d. from a common prior. To learn whether a firm's product matches, the consumer must ``inspect'' it at cost $c>0$, which is a commonly known search or inspection cost. Each firm~$j$, after observing its type, simultaneously announces a subsidy $s_j \in [0,c]$ paid to the consumer upon inspection. The consumer observes the subsidy profile $\mathbf{s} = (s_1,\ldots,s_n)$, updates beliefs about match probabilities, and chooses an inspection order that maximizes expected utility, trading off lower net costs $(c-s_j)$ against inferred quality. We analyze the Perfect Bayesian equilibria of this game and show how subsidized inspection reallocates attention in equilibrium.

Before we describe the analysis of this setting, we relate it to two well-known studies. One way to view our analysis is as a Pandora's Box problem in the spirit of \cite{Weitzman1979} in which the ``boxes'' can play an active role by subsidizing search. In this spirit, the model captures decentralized search in which privately-informed firms try to attract consumers by reducing their inspection costs. A second way to view our analysis is as a platform in the spirit of \cite{Athey-Ellison-11} in which firms go beyond paying for prominence, and instead can subsidize search. After describing our analysis, we highlight how this approach to search on platforms is relevant to the emerging trend of \emph{agentic search}, where consumers delegate discovery and evaluation to AI assistants.%
\footnote{Examples include interfaces such as \textit{ChatGPT}, \textit{Gemini}, and other LLM-based chatbots, as well as Shopify's \textit{Sidekick}, Expedia's conversational trip-planner, and Amazon's \textit{Rufus}.}


Three questions guide our investigation of this directed-search framework. First, how do subsidies signal quality? Second, in this signaling game, do forward induction refinements in the spirit of \cite{cho1987signaling} refine equilibria? Third, how do market fundamentals, competition, search costs, and subsidy prices, affect consumer and social welfare?
How do changes in market fundamentals, more competition, lower search costs, higher subsidy prices, affect consumer and social welfare?

\Cref{sec:sorting} addresses the first question. We establish a \emph{subsidy–sorting principle}: in every equilibrium, higher equilibrium levels of subsidy are rewarded with higher attention, which implies that higher–type firms offer (weakly) higher subsidies. Therefore, subsidies on the equilibrium path are a positive signal of quality.\footnote{In some equilibria, higher \emph{off-path} subsidies could convey a negative signal of value.} 
Given the subsidy-sorting principle, the consumer’s optimal search strategy is therefore to inspect firms in descending order of subsidy, breaking ties randomly. Search stops when a match is found or when the next inspection yields a negative expected payoff.

We next address equilibrium multiplicity and ask which outcomes survive a natural forward-induction refinement. Following \citet{cho1987signaling} and \cite{banks1987equilibrium}, we impose a forward-induction restriction that eliminates equilibria supported by implausible off-path beliefs. Our main result establishes that this refinement selects a unique equilibrium outcome in which every firm type that can profitably induce inspection is inspected. This equilibrium outcome is supported by a a distinctive step–increasing–step  subsidy pattern.\footnote{Our description below and \Cref{fig:intro-example-SIS} uses a subsidy schedule in which all firms whose types are below $\underline{t}$ offer zero subsidies. There exist outcome-equivalent equilibrium that satisfies the refinement in which these firms offer strictly positive subsidies that are nevertheless insufficient to attract the consumer's attention. }

\begin{figure}[h!]
    \centering
    \includegraphics{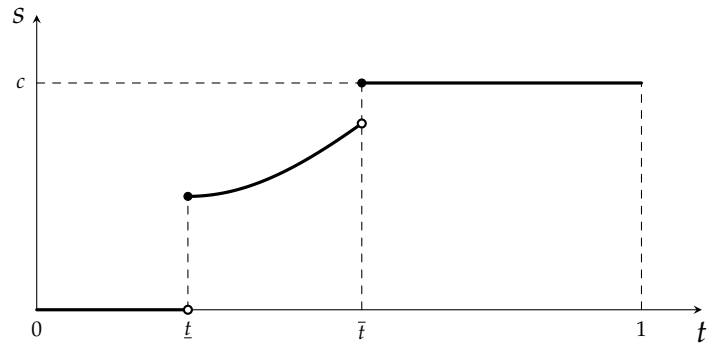}
    \caption{Low types are never inspected, intermediate types separate with strictly increasing subsidies, and high types pool at the full subsidy.}
    \label{fig:intro-example-SIS}
\end{figure}

We illustrate this pattern in \Cref{fig:intro-example-SIS}. 
Low-quality firms offer zero subsidy and are never inspected. Intermediate types separate by offering strictly increasing subsidies, thus revealing their quality. High types pool at the full subsidy, reimbursing the inspection cost entirely. Pooling generates a fundamental tension: it obscures differences among the best firms, yet it makes their inspection free for consumers. The consumer gains fully from this trade–off, but social surplus suffers from the inefficient search due to the loss of information at the top.

\Cref{sec:welfare} addresses the welfare question: how do changes in market primitives affect consumer and social outcomes? As expected, more competing firms or lower search costs increase consumer welfare. Higher per–unit subsidy prices reduce expected transfers and discourage firm participation, both lowering consumer surplus.

We apply these insights to the emerging \emph{agentic economy} \citep{Rothschildetal2025}. AI assistants are reshaping the economics of consumer search.%
\footnote{For example, a shopper might instruct an AI assistant to ``scan every retailer, check return policies, and buy the cheapest 32-inch monitor that ships by Friday''.} %
Unlike human shoppers who rely on salience and limited patience, algorithmic agents follow programmed utility functions and face measurable costs. Each evaluation is a sequence of auditable computational events (e.g., API calls, database queries, or token usage). This transparency creates new levers for market manipulation: sellers can now purchase algorithmic attention directly, subsidizing inspection to push their products ahead of competitors.

\Cref{sec:agentic} addresses the platform question: how should a monopolistic platform price computational attention, and what are the welfare consequences? The platform sets token prices that determine which firms can bid for visibility and how much consumers ultimately search. We show that the platform maximizes profits when search is inefficient. By lowering token prices, it induces pooling among top firms, expands entry, and forces agents to inspect more options than is socially optimal. At the profit–maximizing price, the additional attention and search more than offset the lower per–inspection charge. This strategy raises revenue while appearing to benefit consumers through cheaper individual inspections. In equilibrium, the platform exploits the wedge between consumer welfare and social surplus, extracting rents from the opacity it creates.

\Cref{sec:conclusion} concludes. All omitted proofs are in Appendices. The remainder of this introduction places our contribution within the context of the existing literature.

\paragraph{Relation to the literature.}
The classical search literature treats inspection as a fixed burden borne by consumers. Early work frames search as an optimal stopping problem \citep{Stigler-61,McCall-70,Diamond-71}. \citet{Weitzman1979} provides the canonical index rule when inspection costs are heterogeneous. Subsequent work extends this framework to directed and simultaneous search under match-quality uncertainty \citep{ChadeSmith2006}. We build on this tradition but depart by making search costs contractible: firms can subsidize evaluation, turning a stopping problem into an equilibrium allocation of attention.

A growing strand studies strategic information provision and disclosure in search environments. Firms strategically disclose more when they mainly compete to attract new customers and cannot condition on buyer beliefs; they disclose less when competition is about persuading already-arrived buyers or when firms can condition on beliefs \citep{Board-Lu-18, Au-Whitmeyer-23}. Recent work recasting Weitzman’s Pandora’s box in competitive and intermediary settings finds that competition often improves informativeness, yet principals who sell signals can implement efficient stopping while extracting surplus \citep{Ding_etal-25,Mekonnen_etal-25}. These contributions show that strategic information provision reshapes inspection order, and thus welfare. Our analysis complements this literature but shifts the instrument from beliefs to search frictions: rather than solely manipulating posteriors, firms in our model manipulate the cost of learning. Subsidies lower frictions directly and thereby reorder the search path, benefiting the consumer. 

A second strand examines how firms compete for consumer attention through advertising \citep{Anderson-Renault-06, EdelmanOstrovskySchwarz2007, Athey-Ellison-11, Chen-He-2011}. In sponsored search, firms bid for ad positions, with higher bidders receiving more prominent placement. The generalized second-price auction creates efficient rankings where firms with higher valuations for clicks secure better positions \citep{EdelmanOstrovskySchwarz2007}. \citet{Athey-Ellison-11} show this extends to organic search when firms can influence their ranking through quality investments. We demonstrate that subsidizing inspection achieves similar efficient sorting without auctions or slot constraints. The mechanism is simpler and more direct. Yet we uncover a crucial difference: subsidy caps create pooling at the top, generating excessive search that hurts social welfare while preserving consumer surplus. This distinguishes our setting from position auctions where firms can always separate through higher bids.

The obfuscation literature highlights the opposite strategic move: firms or platforms can raise search frictions to soften competition and extract rents \citep{Ellison-Wolitzky-2012, Nocke-Rey-24}, and recent work shows how algorithmic obfuscation and sponsored slots can be used jointly to boost platform revenue \citep{Janssen-etal-25}. We show how those two levers can be combined: a platform can partly obfuscates organic listings by creating stronger competition among firms. This competition leads to lower search frictions as firms subsidize inspection more aggressively and hit subsidy caps. Consumers perceive inspections as effectively free and search more, while the platform extracts revenue (from firms) through greater search volume. This mechanism nests cleanly in the two-sided platform literature: subsidy pricing is both a participation instrument and an attention-pricing device \citep{RochetTirole2003,Weyl2010}.


\section{Environment}\label{sec:environment}

A consumer seeks a product that matches her needs. The benefit that she accrues from finding such a product is $u>0$; otherwise, she obtains her outside option, $\emptyset$, which delivers zero utility. She searches for the product on a platform, inspecting the products of different firms. Each inspection costs $c>0$ and reveals whether the inspected firm's product meets her needs.%
\footnote{We assume that the inspection cost is not so large that all search is deterred. Formally, $c < (1+pu)/p$, so that some firm types can profitably subsidize inspection and induce the consumer to search.} 
The consumer searches sequentially and terminates either when she finds a match or she opts out.\footnote{We focus on inspection rules that stop either after all planned inspections fail or whenever a match is formed. This is a minimum requirement for optimality in this search environment.}

Each firm $j \in \mathcal{J} = \{1, \dots, n\}$ has a type $t_j \in T\equiv [0,1]$, drawn independently from a common distribution $F$ with full support.\footnote{The full-support assumption is not crucial, but simplifies the exposition. } The interpretation is that \(t_j\) is the firm’s match probability upon inspection.  Hence, we refer to a higher type firm as \emph{better} and a lower type firm as \emph{worse}. A firm earns $v =1$ if its product matches the consumer, and $0$ otherwise. Before inspections occur, a firm may offer a subsidy \(s_j\in[0,c]\) that is paid if the consumer inspects the firm's product. The subsidy reduces the consumer’s effective inspection cost from \(c\) to \(c-s_j\). Each unit of subsidy costs the firm \(p\in\mathbb R_+\) per unit. In \Cref{sec:agentic}, we consider a profit-maximizing choice of $p$ for a platform, but for the baseline model, $p$ is a primitive. 

\paragraph{Timing.}
Each firm privately observes \(t_j\) and simultaneously chooses a subsidy \(s_j\in[0,c]\). After observing the vector of subsidies \(\textbf{s}=(s_j)_{j\in\mathcal J}\), the consumer selects a sequential inspection rule. Upon inspecting firm \(j\), the consumer learns whether the product satisfies her need (i.e., whether a match occurs).

\subsection{Strategies} \label{subsec:strategies}

Each per–unit price \(p\in\mathbb{R}_+\) induces a Bayesian game \(\mathcal{G}(p)\) between the consumer and the set of firms \(\mathcal{J}=\{1,\dots,n\}\). 

A \emph{subsidy policy} is a measurable function $\sigma:T\to[0,c]$, mapping a type \(t\) to a subsidy \(s = \sigma(t)\). A subsidy policy profile $(\sigma_j)_{j \in \mathcal{J}}$ is a collection of subsidy policy stating each firm~$j$ policy. Given a type realization $\textbf{t}$, a subsidy policy profile $(\sigma_j)_{j \in \mathcal{J}}$ induces a realized subsidy profile \(\textbf{s}=(s_j)_{j\in\mathcal J}\) with \(s_j=\sigma_j(t_j)\).

The consumer observes \textbf{s} and then implements an inspection rule. An inspection rule $r \in \mathcal{R}$ is a planned order of inspections with an outside-option stop: inspections proceed sequentially until a match is found or the outside option is reached. Formally, let \(\mathcal R\) be the set of permutations of \(\bar{\mathcal J}=\mathcal J\cup\{\emptyset\}\). For \(r\in\mathcal R\), write \(i\succ_r j\) to mean “\(i\) is scheduled before \(j\),” i.e., \(r(i)>r(j)\); by convention, if \(\emptyset\succ_r j\) then \(j\) is never inspected. The consumer’s strategy is an \emph{inspection policy}, a stochastic kernel \(\iota:\mathcal R\times[0,c]^n\to[0,1]\), where \(\iota(r\mid\textbf{s})\) denotes the probability that inspection rule \(r\) is used at subsidy profile \(\textbf{s}\).

The consumer's \emph{beliefs} are given by a measurable mapping \(\mu:[0,c]^n\to\Delta(T)^n\) that assigns to each subsidy profile \(\textbf{s}\) a posterior over types; our formulation assumes both on- and off-path that the consumer's beliefs about a firm $j$ do not depend on the subsidies offered by other firms.\footnote{Given that firms' types are independent, the stochastic independence of posterior beliefs would be implied by \cite{kreps1982sequential} (were the game to have finitely many actions).} 

An \emph{assessment} is a triple  $ \bigl((\sigma_j)_{j \in \mathcal{J}},\,\iota,\,\mu\bigr)$, where $\sigma_j$ is firm~$j$’s subsidy policy, $\iota$ is the consumer’s inspection policy, and $\mu$ is the consumer’s belief system.

Throughout, ``increasing'' means weakly increasing, i.e.\ $x'\geq x \implies f(x')\geq f(x)$. We use ``strictly increasing'' when the inequality is strict. Analogous definitions apply for decreasing functions and monotone functions.

\subsection{Equilibrium Concept}\label{subsec:equilibrium-concept}

Our solution concept is Perfect Bayesian Equilibrium (PBE). An assessment $((\sigma_j)_{j \in \mathcal{J}}, \iota, \mu)$ is a PBE if (i) each firm’s strategy maximizes expected profit, (ii) the consumer’s strategy maximizes expected utility given beliefs, and (iii) beliefs are updated by Bayes’ rule whenever possible. We restrict attention to symmetric equilibria.

In a symmetric Perfect Bayesian Equilibrium, two restrictions apply. First, \emph{exchangeability}: the consumer treats firms offering identical subsidies identically, both on and off the equilibrium path.%
\footnote{Formally, for any two subsidy profiles $\mathbf{s}$ and $\mathbf{s}'$, if there exists a permutation $\pi$ such that $s_j = s_{\pi(j)}'$ for all $j \in \mathcal{J}$, then $\iota(\mathbf{s}) = \iota(\mathbf{s}')$ and $\mu(\mathbf{s}) = \mu(\mathbf{s}')$.} %
This implies that the belief system can be characterized by a measurable mapping $\mu: [0, c] \rightarrow \Delta(T)$ to which we refer as the consumer's belief function.%
\footnote{Formally, under symmetry and stochastic independence of posteriors, the conditional distribution of a firm's type given that it offers subsidy $s$ is identical across firms and independent of other firms' strategies.} %
Second, all firms adopt the same subsidy policy $\sigma$.
Thus a symmetric assessment can be written as $(\sigma,\iota,\mu)$, where $\sigma$ is the common subsidy policy, $\iota$ the inspection policy, and $\mu$ the consumer's belief function.  Henceforth, we refer to symmetric PBE as \emph{equilibria}.

A key object is the \emph{attention function}. A profile $(\sigma,\iota)$ induces a distribution over inspection paths and outcomes. Offering a subsidy $s$ delivers a level of attention, defined as the ex ante probability that the firm is inspected. Any given firm is inspected if, and only if, all firms scheduled before it fail to match. To encode the outside option, define $t_{\emptyset}:=1$. The attention function of a firm offering $s$ under $(\sigma,\iota)$ is
\begin{equation}\label{eq:q-sigma}
q_{\sigma,\iota}(s)
\equiv
\int \Bigg(
\sum_{r\in\mathcal R}
\iota\big(r\mid (s,\sigma(\mathbf{t}_{-j}))\big)\,
\prod_{i\in\bar{\mathcal J}:\, j\prec_r i} (1-t_i)
 \Bigg)\, dF^{n-1}(\mathbf{t}_{-j}).
\end{equation}
In words, the \emph{attention function} at $s$, $q_{\sigma,\iota}(s)$, is the ex ante probability that a firm offering subsidy $s$ is inspected under equilibrium play.

From an interim perspective, a type–\(t\) firm that offers subsidy \(s\) is inspected with probability \(q_{\sigma,\iota}(s)\), succeeds in matching with probability \(t\), and pays \(p\,s\) upon inspection. The firm's \textit{interim} profit is therefore
\begin{equation}\label{eq:pi}
\pi_{\sigma,\iota}(s;t) = \big(t - p\,s\big)\, q_{\sigma,\iota}(s).
\end{equation}

Turning to the consumer, a match fails to occur under rule \(r\) precisely when every firm scheduled before \(\emptyset\) fails to match. Thus, given a posterior belief $\mu(\textbf{s})$ over types, the consumer expects to match on expectation $ \mathbb{E}_{\textbf{t}\sim\mu(\textbf{s})}\!\left[1 -
\prod_{i\in \bar{\mathcal J}:\, i\prec_r \emptyset} (1-t_i) \right]$.



\begin{definition}
An equilibrium of the game $\mathcal{G}(p)$ is a triple $(\sigma,\iota,\mu)$ such that:
\begin{enumerate}[label=(\roman*)]
    \item \textbf{Firm optimality.} Given $(\sigma,\iota)$, each type-$t$ firm chooses
    \[
    \sigma(t) \in \arg\max_{s\in[0,c]} \big(t - p s\big)\, q_{\sigma,\iota}(s).
    \]
    
    \item \textbf{Consumer optimality.} 
    Given a realized subsidy profile $\mathbf{s}$ and posterior beliefs $\mu(\mathbf{s})$, the inspection distribution $\iota(\mathbf{s}) \in \Delta(\mathcal{R})$ assigns positive probability only to rules that maximize expected utility:
    \[
    \mathrm{supp}\big(\iota(\mathbf{s})\big) \;\subseteq\; 
    \arg\max_{r \in \mathcal{R}} 
    \mathbb{E}_{\mathbf{t} \sim \mu(\mathbf{s})}
    \left[
    u \left(1 - \prod_{i\in \bar{\mathcal J}:\, i\prec_r \emptyset} (1-t_i)\right) 
    - \sum_{j \in \mathcal{J}} (c-s_j)\, \prod_{i\in \bar{\mathcal J}:\, i\prec_r j} (1-t_i)
    \right].
    \]

    \item \textbf{Belief consistency.} $\mu$ is derived from $\sigma$ by Bayes’ rule wherever possible.
\end{enumerate}
\end{definition}

Equilibria always exist. For example, if all firms set zero subsidies, the consumer searches randomly on the equilibrium path and assigns off-path beliefs that any subsidizing firm is of type $0$, no such firm is inspected. In \Cref{sec:consumer-optimal-eqm} we construct a consumer--optimal equilibrium, which applies a forward--induction logic in the spirit of \citet{cho1987signaling} and rules out reliance on implausible off--path beliefs.

\subsection{Equilibrium Refinement}\label{subsec:refinement-intro}


Following \citet{cho1987signaling} and \citet{banks1987equilibrium}, we refine equilibria at off-path beliefs. We require that at off-path subsidy $s$ the consumer do not place positive probability in \emph{dominated} types. The idea is as follows: after an off–path subsidy, the consumer assigns zero probability to type $t$ if there exists some other type $t'$ for whom that subsidy yield a higher payoff than the the equilibrium payoff under any  posterior belief and its corresponding inspection policy. Thus, when evaluating off-path moves, the consumer rules out firm types for whom that deviation is not dominant.

Fix an equilibrium $(\sigma,\iota,\mu)$ of $\mathcal{G}(p)$. A subsidy $s\in[0,c]$ is \emph{off–path} if there is not type-$t$ that offers it, that is, $s\notin \operatorname{Im}(\sigma)$. Let $\mathcal{C}(\sigma)$ denote the set of admissible consumer's continuations based on $\sigma$: each element is a pair $(\iota',\mu')$ such that (i) the consumer’s inspection policy $\iota'$ is optimal given beliefs $\mu'$, (ii) exchangeability holds, and (iii) on–path beliefs obey Bayes’ rule. Given $(\iota',\mu')\in\mathcal{C}(\sigma)$, the expected profit for a type-$t$ firm that offers $s$ is $\pi_{\sigma, \iota'}(s;t) \;=\; \big(t - p s\big)\, q_{\sigma,\iota'}(s)$.

\begin{definition}
For an off-path subsidy level $s$, we say that $(t, s)$ is \emph{dominated by the equilibrium value} if 
\[ \pi_{\sigma, \iota}(\sigma(t); t) > \max_{(\iota', \mu') \in \mathcal{C}(s)}\pi_{\sigma, \iota'}(s; t). \]
That is, type-$t$ cannot improve its equilibrium payoff by deviating to the off-path subsidy~$s$ under any admissible continuation.

Moreover, we say that type $t$ is \emph{dominated} by type $t'$ at $s$ if:
\[ \pi_{\sigma,\iota}(\sigma(t);t) \leq \pi_{\sigma,\iota'}(s;t) \implies \pi_{\sigma,\iota}(\sigma(t');t') < \pi_{\sigma,\iota'}(s;t')
    \quad \forall (\iota', \mu') \in \mathcal{C}(\sigma)
\]
Let $t \in D_{\sigma, \iota}(s)$ if $(t,s)$ is dominated by the equilibrium value or there exits some $t'$ such that $t$ is dominated by $t'$ at $s$ .
\end{definition}

\begin{definition}
An equilibrium $(\sigma,\iota,\mu)$ fails the equilibrium refinement if, there exists some off–path subsidy $s$, and a type $t$ such that $t \in \textrm{supp}(\mu(s))$ and $t \in D_{\sigma, \iota}(s)$.%
\footnote{Abusing notation, we write \(\mu(s)\) for the marginal distribution of the first coordinate of \(\mu\) evaluated at \(s\). By symmetry this marginal is the same for every firm, so we omit any firm index.}
\end{definition}

Operationally, equilibrium-dominance restricts off–path beliefs to the \emph{economically credible} types for whom the deviation could be optimal. In \cref{sec:consumer-optimal-eqm}, we show that there is a unique equilibrium outcome that survives this refinement and characterize the set of equilibria that achieves it.

\section{The Subsidy–Sorting Principle}\label{sec:sorting}

Do high–quality firms subsidize more than low–quality ones, or can competition overturn this relationship? This section establishes the \emph{Subsidy–Sorting Principle}. In every equilibrium, the attention a firm receives is increasing in its subsidy. This monotonicity implies that equilibrium subsidy policies $\sigma$ are increasing in type, and that the consumer follows a descending–subsidy inspection rule: firms are visited in descending order of their realized subsidies, with uniform tie–breaking within pools. Search terminates either upon a successful match or when the next inspection yields negative expected surplus.

\begin{theorem}[Subsidy–Sorting Principle]\label{thm:sorting}
In every PBE of the game $\mathcal{G}(p)$, a higher subsidy generates more attention. Moreover,
\begin{enumerate}
\item \textbf{Monotone subsidies:}  the subsidy policy is increasing in type, i.e., if $t'\geq t$ then $\sigma(t')\ge \sigma(t)$.
\item \textbf{Descending–subsidy index rule (DSIR):}
the consumer inspects firms in descending order of their realized subsidies, breaking ties uniformly. After a failed inspection, she continues search only if for the next inspected subsidy level $s$ the expected benefit is higher than the expected cost, i.e. $\mathbb{E}[t \mid \sigma(t) = s] \geq \frac{c -s}{u}$.
\end{enumerate}
\end{theorem}

The above result offers several structural properties that must be satisfied by any equilibrium. Higher subsidies draw more inspections. Better types exploit this by offering larger subsidies. The consumer then inspects in descending order of subsidies, breaking ties at random, and stops once further search is not worthwhile.

The reasoning unfolds in two steps. First, inspection probabilities are increasing in subsidy, so subsidies sort firms by how much attention they receive. Second, because firms' profits have a single-crossing property, higher types benefit more from extra attention and therefore choose higher subsidies. Once subsidies are monotone, descending-subsidy search is optimal: higher subsidies both lower net costs and signal higher chances of a match, and the consumer stops when the expected gain turns negative.

\subsection{Subsidy Monotonicity}

Fix an equilibrium. We argue that offering a higher subsidy must increase the (ex-ante) inspection probability. Consider two on-path subsidies $s<s'$, i.e., there exists $t$ and $t'$ such that $\sigma(t)=s$ and $\sigma(t')=s'$. We claim that $q(s')$, the ex-ante probability that the consumer inspects given subsidy $s'$, is higher than $q(s)$, the corresponding probability given subsidy $s$. Towards a contradiction, suppose that $q(s')< q(s)$. Then, type $t'$ payoff from deviating to a subsidy $s$ would be strictly higher than its equilibrium payoff:
\begin{align*}
    (t'-p\,s)\, q(s) > (t'-ps) q(s') \geq (t'-ps') q(s'),
\end{align*}
where the first inequality follows from $q(s) > q(s')$, and the second from $s< s'$.\footnote{The second inequality is not strict as nothing prevents $q(s') = 0$.} Therefore, type $t'$ would have a strictly profitable deviation. Putting it simply, if a firm sacrificed attention by subsidizing search, then it would have no interest in doing so. 

The second step uses monotone comparative statics to show that better firms have a greater interest in subsidizing search. As attention is more rewarding for higher types, they have a stronger interest in obtaining it. Specifically, each firm's payoff over inspection probabilities satisfies single-crossing differences \citep{MilgromShannon1994} in type. Suppose that type $t$ (resp., $t'$) offers subsidy $s$ (resp., $s'$) and obtains an inspection probability $q$ (resp., $q'$). We argue that if $t'\geq t$, then $q'\geq q$ by evaluating each type's best response. In equilibrium, type $t'$ must favor subsidy $s'$ over $s$, which implies
\begin{align*}
    (t'-ps')q' \geq (t'-ps)q. 
\end{align*}
Analogously, type $t$ must favor subsidy $s$ over $s'$, which implies
\begin{align*}
    (t-ps)q \geq (t-ps')q'. 
\end{align*}
Rearranging the above inequalities yields
\begin{align*}
    t'(q'-q) \geq p(q's'-qs) \geq t (q'-q). 
\end{align*}
Now, if it were the case that $q'<q$, the above inequality would imply that $t'<t$, which results in a contradiction.

Pasting these two steps together, we obtain that: (i) the consumer pays more attention to firms that offer higher subsidies, and (ii) better firms have more to gain from the consumer's attention. This establishes that, in every equilibrium, higher types must offer greater subsidies and be rewarded with more attention. 

\begin{lemma}\label{lem:subsidy-policy-weakly-increasing}
    In any equilibrium $(\sigma,\iota,\rho)$, the inspection probability $q_{\sigma,\iota}$ is increasing in on-path subsidies, and the subsidy policy $\sigma$ is increasing in type.
\end{lemma}

\subsection{Consumer Search: Ordering and Stopping}

Having established that equilibrium subsidies are weakly increasing in type, we now turn to the consumer’s problem. Higher types offer higher subsidies and, as a result, receive more attention. This suggests that, in expectation, firms with larger subsidies should be inspected earlier. We show that this intuition is correct: the consumer’s optimal policy is to inspect in descending order of a simple index that is monotone in subsidies on the equilibrium path.

Given a subsidy policy $\sigma$ and a belief system $\mu:[0,c]\to\Delta(T)$, a realized subsidy $s$ reveals that the issuing firm’s type lies in the level set $\sigma^{-1}(s)$. Because the consumer’s payoff is linear in the match indicator, all payoff–relevant information is summarized by the \emph{posterior success probability}
\begin{equation}\label{eq:post-success-prob}
    \tau_\mu(s) \;:=\; \mathbb{E}_{t\sim\mu(s)}[t].
\end{equation}
When $s\in S_{\sigma}:=\{\sigma(t):t\in T\}$ and $\mu$ is Bayes consistent, $\tau_\mu(s)=\mathbb{E}[t\mid \sigma(t)=s]$. If $s\notin\mathrm{Im}(\sigma)$, $\tau_\mu(s)$ is determined by off–path beliefs.

Fix a realized subsidy profile $\mathbf{s}=(s_j)_{j\in\mathcal{J}}$ with $s_j\in S_\sigma$ for all $j$. Write $\tau_j:=\tau_\mu(s_j)$ and $\kappa_j:=c-s_j$ for the posterior match probability and the net inspection cost. Inspecting firm~\(j\) yields a match with probability $\tau_j$ and costs $\kappa_j$. The expected payoff is therefore $\tau_j(u - \kappa_j) + (1-\tau_j)(V - \kappa_j)$, while stopping delivers $V$. Inspection is optimal when the former exceeds $V$, which defines the reservation index
\begin{equation} \label{eq:reservation-index}
    r(\tau, \kappa) := u - \frac{\kappa}{\tau},
\end{equation}
with the convention $r(\tau,\kappa)=-\infty$ when $\tau=0$. Inspection is optimal if and only if $V < r(\tau,\kappa)$. This is exactly the index rule in \citet{Weitzman1979} for a Bernoulli prize with success probability $\tau$ and cost $\kappa$.

We introduce the \emph{descending–subsidy index rule} under beliefs $\mu$.

\begin{definition}[Descending–subsidy index rule]\label{def:descending-subsidy-rule-policy}
Fix $\mu:[0,c]\to\Delta(T)$ and define
\begin{equation}
    r_\mu(s)\;:=\;u-\frac{c-s}{\tau_\mu(s)},
\end{equation}
with $r_\mu(s)=-\infty$ if $\tau_\mu(s)=0$. The \emph{descending–subsidy index rule} (DSIR$_\mu$) prescribes:
\begin{enumerate}[label=(\alph*),nosep]
    \item \textbf{Ordering:} inspect firms in descending order of $r_\mu(s)$, breaking ties uniformly at random;
    \item \textbf{Stopping:} terminate upon the first match, or whenever the next available index is nonpositive.
\end{enumerate}
The search policy that implements DSIR$_\mu$ for any realized subsidy profile is the \emph{descending–subsidy index policy}, DSIP$_\mu$.
\end{definition}

The next result states that the DSIP is consumer optimal. Moreover, it is unique under exchangeability. Note that the result do not requires $S_\sigma=[0,c]$. If a realized subsidy $s\notin S_\sigma$ is observed, DSIR$_\mu$ remains well defined once $\tau_\mu(s)$ is specified by off–path beliefs; different off–path beliefs can support different equilibria.

\begin{lemma}\label{lem:descending-search}
For any belief system $\mu$, DSIP$_\mu$ maximizes the consumer’s expected utility. It is unique up to tie–breaking among firms with equal index.
\end{lemma}

Based on this insights, we can easily describe the consumer search patterns on equilibrium. Whenever the subsidy policy is increasing, the consumer should always prefer inspecting a higher–subsidy firm: it offers both a higher probability of match and a lower net cost. If the consumer did not inspect such a firm before a lower–subsidy one, she could strictly improve her payoff by reordering. Once a match is found, she has no incentive to continue, since further inspections only add cost without raising her payoffs. And if no match has been found, inspection should stop once the next available firm yields negative expected surplus, since the outside option guarantees a payoff of zero. This logic pins down the consumer’s optimal rule: inspect in descending order of subsidies, break ties at random, stop at the first match, and otherwise continue only while the expected net benefit of the next inspection is positive. 

\begin{corollary}\label{cor:sorting-on-path}
Suppose $\sigma$ is weakly increasing and $\mu$ is Bayes consistent on $S_\sigma$. Then $r_\mu(\cdot)$ is increasing on $S_\sigma$. Consequently, DSIR$_\mu$ reduces to inspecting firms in weakly descending order of realized subsidies, breaking ties uniformly, and continuing search only while $ \mathbb{E}\big[t\mid \sigma(t)=s\big]\;>\;\frac{c-s}{u}$.
\end{corollary}

Combining \cref{lem:subsidy-policy-weakly-increasing} with \cref{cor:sorting-on-path} immediately yields the Subsidy–Sorting Principle: higher types offer weakly higher subsidies; the consumer inspects in descending subsidy order; and she stops exactly when she matches or when the next firm expected profits are lower than the net costs.

\section{Equilibrium Uniqueness Under Refinement}\label{sec:consumer-optimal-eqm}

We now characterize the unique equilibrium outcome that survives forward–induction refinement. We refer to the equilibrium that induces it the \emph{reasonable equilibrium}. This equilibrium  takes a simple shape: low types do not subsidize, intermediate types separate with a strictly increasing schedule, and high types pool at the cap $c$. This equilibrium induces efficient participation: every firm that could profitably signal so as to secure positive inspection does so and is inspected with positive probability. The formal result is stated in \cref{thm:uniqueness}.

Our construction proceeds in four steps. First, we derive firm incentives under full separation, yielding a differential equation that characterizes truth-telling subsidies. Second, we identify the lowest participating type, which pins down the admissible boundary condition. Third, if the implied policy exceeds the subsidy cap $c$, we iron back the schedule by pooling higher types at the cap. Finally, we establish that this outcome is the only one consistent with refinement.

\subsection{Firm Incentives under Full Separation}

Assume a strictly increasing subsidy policy $\sigma$. Under full separation, subsidies reveal types one-to-one. By \cref{lem:descending-search}, the consumer should inspect firms in descending order of subsidies, breaking ties uniformly, and stops once the next inspection fails the reservation test. Let $\underline{t}$ denote the lowest type inspected.\footnote{Under full separation, the subsidy perfectly reveals $t$, so the reservation index is $r(t)=u-(c-\sigma(t))/t$. Inspection requires $r(t)\ge 0$, hence the lowest inspected type solves $\sigma(\underline{t})=c-u\underline{t}$. Whenever $\sigma(1) > c-u$, the solution exists and is unique, with $\underline{t}\in(0,c)$. Otherwise, no type is inspected under this policy, set $\underline{t} = 1$.} Firms with $t<\underline{t}$ are never inspected.

Consider the incentives of a type-$t$ firm. If it mimics some type $\tilde t$, it offers subsidy $\sigma(\tilde t)$ and is inspected with probability $q_n^{\text{sep}}(\tilde t;\underline{t})$. The key property of this inspection probability is that, under full separation, it depends only on $\underline{t}$ and not on the detailed shape of $\sigma$.\footnote{A closed form is derived in Appendix~\ref{app:inspection-prob}, where we show 
\(
q_n^{\text{sep}}(t;\underline{t})=(1-\int_t^1 x\,dF(x))^{n-1}\mathbf{1}\{t\ge \underline{t}\}.
\)}  
The firm’s expected payoff from mimicking $\tilde t$ is
\begin{equation}\label{eq:firm-profit}
    \pi(\tilde t;t) \;=\; \bigl[t - p\,\sigma(\tilde t)\bigr]\,q_n^{\text{sep}}(\tilde t; \underline{t}).
\end{equation}

In equilibrium, truth-telling must maximize \eqref{eq:firm-profit}. Since $\sigma$ is strictly increasing and hence almost everywhere differentiable, the envelope condition at $\tilde t=t > \underline{t}$ implies
\[
   -\,p\,\sigma'(t)\,q_n^{\text{sep}}(t; \underline{t})
   \;+\;
   \bigl[t - p\,\sigma(t)\bigr]\,(q_n^{\text{sep}})'(t; \underline{t})
   \;=\;0,
\]
or equivalently that marginal cost should be equal to marginal benefit:
\begin{equation}\label{eq:ODE-raw}
    p\frac{d}{dt}\!\big[\sigma(t)\,q_n^{\text{sep}}(t; \underline{t})\big]
    \;=\;
    t\,(q_n^{\text{sep}})'(t; \underline{t}).
\end{equation}
Every strictly increasing, truth-telling subsidy policy must satisfy \eqref{eq:ODE-raw}.

\subsection{Boundary Condition from Participation and Inspect--Worthiness}

To integrate \eqref{eq:ODE-raw}, we require a boundary condition. The relevant support of the subsidy policy begins at the lowest type that both finds participation profitable and is just worth inspecting. These two primitives jointly determine the boundary.  

A firm that expects to be inspected with positive probability participates only if it expects nonnegative profit, i.e.\ $\sigma(t)\le t/p$.  
The consumer inspects only if the reservation index is nonnegative, which under separation requires $\sigma(t)\ge c-t$.  
Together these inequalities define a feasible corridor where trade is efficient:
\[
  \max\{0,\,c-t\}\;\le\;\sigma(t)\;\le\;\frac{t}{p}.
\]
Inefficient trade means that some party, either the consumer or the firm, expected negative payoff of such transaction. \cref{fig:no-trade-region} illustrates such region. 

\begin{figure}[t!]
    \centering
    \begin{tikzpicture}[>=stealth, line cap=round, line join=round, x=6cm, y=4cm]
  \pgfmathsetmacro{\c}{0.5}     
  \pgfmathsetmacro{\p}{1.0}     

  \pgfmathsetmacro{\tI}{min(1,\p*\c)}    
  \pgfmathsetmacro{\tc}{min(1,\c)}       

  \draw[->] (0,0) -- (1.10,0) node[below] {$t$};
  \draw[->] (0,0) -- (0,\c+0.15) node[left] {$s$};

  \draw[dashed] (0,\c) -- (1,\c);
  \draw[dashed] (1,0) -- (1,\c);

  \node[left] at (0,\c) {\scriptsize $c$};
  \node[below] at (\c,0) {\scriptsize $c$};
  \node[below] at (\c/2,0) {\scriptsize $\underline{t} = \frac{pc}{1+p}$};
  \node[left] at (0,\c/2) {\scriptsize $\underline{s} = \frac{c}{1+p}$};

  \draw[very thick,gray!70!black]
    plot[domain=0:\tI] (\x,{\x/\p});
  \draw[very thick,gray!70!black,dashed]
    plot[domain=\tI:\c+0.05] (\x,{\x/\p}) ;

    \draw[gray]
      plot[domain=0:\tc] (\x,{\c-\x}) ;
    
    \path[fill=red!40,draw=red,pattern=north west lines,pattern color=red]
       (0,0) -- (0,\c) -- (\tc,0) -- cycle;

  \path[fill=blue!35,draw=blue!70!black,pattern=north east lines,pattern color=blue!70!black]
     (0,\c) -- (\tI,\c) -- (\tI,{\tI/\p}) -- (0,0) -- cycle;

  \node[below] at (0,0) {\scriptsize 0};
  \node[below] at (1,0) {\scriptsize 1};



    \filldraw[black] ( \c/2 ,\c/2) circle (1.5pt);


    \draw[dashed, black] (0,\c/2) -- (\c/2,\c/2);
    \draw[dashed, black] (\c/2,0) -- (\c/2,\c/2);

\end{tikzpicture}
    \caption{The shaded regions mark subsidy levels that make trade infeasible. In red, inspection yields negative payoff for the consumer (net costs exceed expected gains). In blue, firms incur losses (expected sales do not cover the subsidy). The dot denotes the cutoff pair $(\underline{t},\underline{s})$, the highest low type that leaves both sides exactly indifferent. In this figure: $c = 0.5$ and $p = 1$.}
    \label{fig:no-trade-region}
\end{figure}
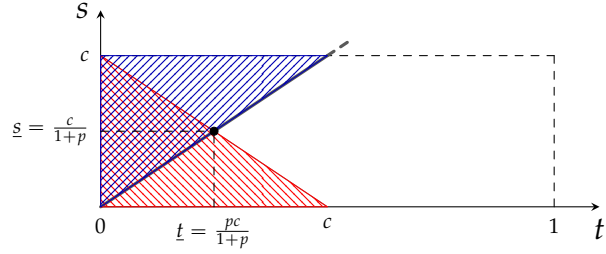

The two bounds meet at a unique cutoff $\underline{t}$. This cutoff plays a dual role: it is the marginal participant, breaking even on subsidies and exactly at the threshold of inspect--worthiness:
\[
  \underline{t}(p,c) \;=\; \frac{pc}{1+p} \in (0,1], 
  \qquad \underline{s} \;=\; \frac{\underline{t}(p,c)}{p} \;=\; \frac{c}{1+p}.
\]

Solving \cref{eq:ODE-raw} with boundary condition $(\underline{t},\,\underline{s})$ yields a strictly increasing, incentive–compatible subsidy policy for $t \in [\underline{t},1]$:
\begin{equation}\label{eq:sigma-rev}
  \sigma^{\mathrm{sep}}(t; \underline{t})
  \;=\; \frac{t}{p} 
  - \frac{\int_{t_0}^t q_n^{\mathrm{sep}}(x;\underline{t})\,dx}{p\,q_n^{\mathrm{sep}}(t; \underline{t})},
\end{equation}
with $\sigma^{\mathrm{sep}}(t; \underline{t})=0$ for $t<\underline{t}$.  
We refer to \cref{eq:sigma-rev} as the \emph{separating policy} with support on $[\underline{t}, 1]$, since it arises under full separation when subsidies perfectly sort types.

\subsection{Ironing Back to the Original Constraints}

The revealing schedule $\sigma^{\text{sep}}(\cdot; \underline{t})$ solves the firm’s problem under full separation but ignores the statutory cap $\sigma(t)\le c$. Restoring the cap leaves two cases. If $\sigma^{\text{sep}}(1;\underline{t})\le c$, then $\sigma^{\text{sep}}(\cdot; \underline{t})$ is already feasible. If instead $\sigma^{\text{sep}}(1;\underline{t})>c$, the upper tail must be truncated and high types pooled at the maximum subsidy $s=c$.  

\begin{figure}[h!]
    \centering
    \begin{subfigure}{0.48\textwidth}
        \centering
        \resizebox{\linewidth}{!}{\begin{tikzpicture}[>=stealth, line cap=round, line join=round, x=8cm, y=6cm]
  \pgfmathsetmacro{\c}{0.5}     
  \pgfmathsetmacro{\p}{1.8}     

  \pgfmathsetmacro{\tI}{min(1,\p*\c)}    
  \pgfmathsetmacro{\tc}{min(1,\c)}       
  \pgfmathsetmacro{\to}{(\p*\c)/(1 + \p)} 
  
  \draw[->] (0,0) -- (1.10,0) node[below] {$t$};
  \draw[->] (0,0) -- (0,\c+0.15) node[left] {$s$};

  \draw[dashed] (0,\c) -- (1,\c);
  \draw[dashed] (1,0) -- (1,\c);

  \node[left] at (0,\c) {\scriptsize $c$};
  \node[below] at (\c,0) {\scriptsize $c$};

  \draw[very thick,gray!70!black]
    plot[domain=0:\tI] (\x,{\x/\p});
  \draw[very thick,gray!70!black,dashed]
    plot[domain=\tI:\c+0.05] (\x,{\x/\p}) ;

    \draw[gray]
      plot[domain=0:\tc] (\x,{\c-\x}) ;
    
    \path[fill=red!40,draw=red,pattern=north west lines,pattern color=red]
       (0,0) -- (0,\c) -- (\tc,0) -- cycle;

  \path[fill=blue!35,draw=blue!70!black,pattern=north east lines,pattern color=blue!70!black]
     (0,\c) -- (\tI,\c) -- (\tI,{\tI/\p}) -- (0,0) -- cycle;

  \node[below] at (0,0) {\scriptsize 0};
  \node[below] at (1,0) {\scriptsize 1};

    \draw[very thick, black, smooth] plot file {sigma_data2.txt};

    \draw[very thick, black] (0,0) -- (\to,\to/\p);
    \filldraw[black] ( \to, \to/\p) circle (1.5pt);

\end{tikzpicture}}
        \caption{Feasible without cap: if $\sigma^{\text{sep}}(1;\underline{t}) \le c$, the revealing schedule requires no adjustment.}
        \label{fig:no-need-cap}
    \end{subfigure}\hfill
    \begin{subfigure}{0.48\textwidth}
        \centering
        \resizebox{\linewidth}{!}{\begin{tikzpicture}[>=stealth, line cap=round, line join=round, x=8cm, y=6cm]
  \pgfmathsetmacro{\c}{0.5}     
  \pgfmathsetmacro{\p}{1.0}     

  \pgfmathsetmacro{\tI}{min(1,\p*\c)}    
  \pgfmathsetmacro{\tc}{min(1,\c)}       
  \pgfmathsetmacro{\to}{(\p*\c)/(1 + \p)} 

  \draw[->] (0,0) -- (1.10,0) node[below] {$t$};
  \draw[->] (0,0) -- (0,\c+0.15) node[left] {$s$};

  \draw[dashed] (0,\c) -- (1,\c);
  \draw[dashed] (1,0) -- (1,\c);

  \node[left] at (0,\c) {\scriptsize $c$};
  \node[below] at (\c,0) {\scriptsize $c$};

  \draw[very thick,gray!70!black]
    plot[domain=0:\tI] (\x,{\x/\p});
  \draw[very thick,gray!70!black,dashed]
    plot[domain=\tI:\c+0.05] (\x,{\x/\p}) ;

    \draw[gray]
      plot[domain=0:\tc] (\x,{\c-\x}) ;
    
    \path[fill=red!40,draw=red,pattern=north west lines,pattern color=red]
       (0,0) -- (0,\c) -- (\tc,0) -- cycle;

  \path[fill=blue!35,draw=blue!70!black,pattern=north east lines,pattern color=blue!70!black]
     (0,\c) -- (\tI,\c) -- (\tI,{\tI/\p}) -- (0,0) -- cycle;

  \node[below] at (0,0) {\scriptsize 0};
  \node[below] at (1,0) {\scriptsize 1};

    \draw[very thick, black, smooth] plot file {sigma_data.txt};

    \draw[very thick, black] (0,0) -- (\to,\to/\p);
    \filldraw[black] ( \to, \to/\p) circle (1.5pt);

\end{tikzpicture}}
        \caption{Binding cap: if $\sigma^{\text{sep}}(1;\underline{t}) > c$, the upper branch needs to be truncated so that high types pool at $s=c$.}
        \label{fig:need-cap}
    \end{subfigure}
    \caption{Optimal separating subsidies. \\ \, Parameters: $F(x)=x$, $c=0.5$, $n=10$. Panel (a): $p=1.8$. Panel (b): $p=1$.}

    \label{fig:cap-comparison}
\end{figure}
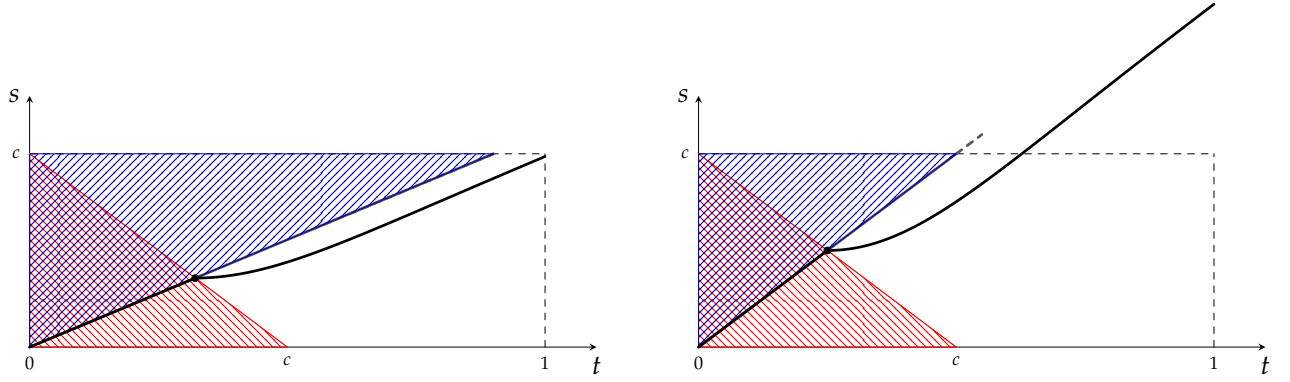

Pooling at $c$ creates a trade–off. On the one hand, inspection becomes costless, so any deviation that offers $s<c$ is strictly dominated, even if it claims to be the highest type, and is inspected only after free inspection firms are exhausted. On the other hand, congestion dilutes attention: as more types join the pool, tie–breaking reduce each firm’s inspection probability.

Let $q^{\mathrm{pool}}(x)\in(0,1]$ denote the inspection probability of a firm offering the full subsidy $s=c$ when all types in $[x,1]$ do so.%
\footnote{Formally,
\[
q^{\mathrm{pool}}(x) 
= \mathbb{E}_{K}\!\left[\frac{1-(1-\mathbb{E}[t\mid t\ge x])^{K+1}}{(K+1)\,\mathbb{E}[t\mid t\ge x]}\right],
\]
with $K\sim \mathrm{Bin}(n-1,1-F(x))$. The derivation is given in \cref{app:pooling}.} %
In equilibrium, the pooling cutoff $\bar t$ is determined by two conditions. 
First, no type $t>\bar t$ should prefer to deviate from pooling (\textbf{upper IC}). 
Second, the marginal type $\bar t$ must be indifferent between pooling and revealing (\textbf{boundary indifference}). 
These conditions are
\begin{align}
[t-pc]\,q^{\mathrm{pool}}(\bar t) &\;\ge\; \sup_{s\in[0,c]} [t-ps]\,q^{\mathrm{sep}}(\sigma^{-1}(s; \underline{t})), 
&&\forall t\in[\bar t,1], \tag{Upper IC}\label{eq:upperIC}\\
[\,\bar t-pc\,]\,q^{\mathrm{pool}}(\bar t) &\;=\; [\,t_1-p\,\sigma^{\textrm{sep}}(\bar t;\underline{t})\,]\,q^{\mathrm{sep}}(\bar t).
\tag{Boundary indifference}\label{eq:boundary}
\end{align}
Here $\sigma^{-1}(\cdot;\underline{t})$ denotes the inverse of the separating schedule $\sigma^{\mathrm{sep}}(\cdot;\underline{t})$.

The \ref{eq:boundary} condition pin down the marginal type $ \bar t$. \cref{lem:existence-cutoff} in \cref{app:pooling} shows that if $\sigma^{\textrm{sep}}(1; \underline{t})>c$, then a unique cutoff $\bar t$ always exists. The upper IC constraint follows from the increasing-differences property of the firm's payoff.  This cutoff is at the left of the type $t_c$ such that $\sigma^{\text{sep}}(t_c; t_0) = c$ generating a discontinuity in the subsidy policy.

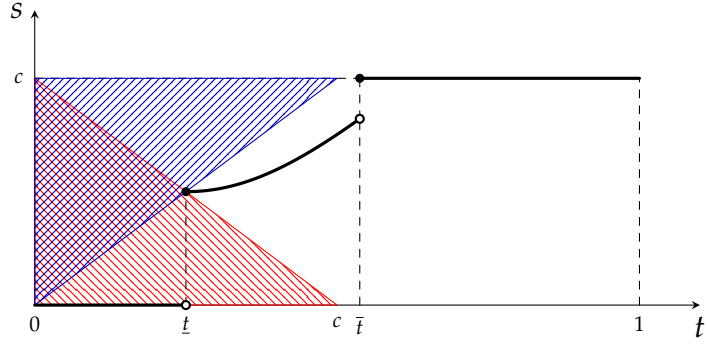
\begin{figure}[h!]
    \centering
    \begin{tikzpicture}[>=stealth, line cap=round, line join=round, x=8cm, y=6cm]
  \pgfmathsetmacro{\c}{0.5}     
  \pgfmathsetmacro{\p}{1.0}     
  \pgfmathsetmacro{\tI}{min(1,\p*\c)}    
  \pgfmathsetmacro{\tc}{min(1,\c)}       
  \pgfmathsetmacro{\to}{(\p*\c)/(1 + \p)} 
  \draw[->] (0,0) -- (1.10,0) node[below] {$t$};
  \draw[->] (0,0) -- (0,\c+0.15) node[left] {$s$};
  \draw[dashed] (0,\c) -- (1,\c);
  \draw[dashed] (1,0) -- (1,\c);
  \node[left] at (0,\c) {\scriptsize $c$};
  \node[below] at (\c,0) {\scriptsize $c$};
  \node[below] at (\to,0) {\scriptsize $\underline{t}$};
  \node[below] at (0.5375,0) {\scriptsize $\overline{t}$};
  
    
    \path[fill=red!40,draw=red,pattern=north west lines,pattern color=red]
       (0,0) -- (0,\c) -- (\tc,0) -- cycle;
  \path[fill=blue!35,draw=blue!70!black,pattern=north east lines,pattern color=blue!70!black]
     (0,\c) -- (\tI,\c) -- (\tI,{\tI/\p}) -- (0,0) -- cycle;
  \node[below] at (0,0) {\scriptsize 0};
  \node[below] at (1,0) {\scriptsize 1};
    \draw[very thick, black, smooth] plot file {sigma_data_trun.txt};
    \draw[very thick, black] (0,0) -- (\to,0);
    
    \filldraw[black] ( \to, \to/\p) circle (1.5pt);
   
    \filldraw[black] ( 0.5375, \c) circle (1.5pt);
    \draw[very thick, black] (0.5375, \c) -- (1,\c);
    \draw[black, dashed] (0.5375, 0.41075555) -- (0.5375, \c) ;
    \draw[black, dashed] (0.5375, 0) -- (0.5375, 0.41075555) ;
    \draw[black, dashed] (\to,0) -- (\to, \to/\p) ;

     \draw[thick, black, fill=white] (0.5375, 0.41075555) circle (1.5pt);
     \draw[thick, black, fill=white] (\to, 0) circle (1.5pt);
\end{tikzpicture}
    \caption{Example of a SIS subsidy policy}
    \label{fig:consumer-optimal}
\end{figure}

We now introduce a family of subsidy policies with a step--increasing--step structure. Each member is indexed by the lowest type that is inspected with positive probability.  

\begin{definition}[Step--increasing--step family]\label{def:optimal-SIS-family}
For any $t_0 \in [\underline{t}, \max\{cp, 1\}]$, define the subsidy policy $\sigma^{*}(\cdot\,;t_0):T\to[0,c]$ by
\[
  \sigma^{*}(t;\underline{t}) \;=\;
  \begin{cases}
    0, & t < t_0, \\[6pt]
    \dfrac{t}{p} \;-\; \dfrac{\int_{t_0}^t q_n^{\mathrm{sep}}(x;t_0)\,dx}{p\,q_n^{\mathrm{sep}}(t;t_0)}, & t \in [t_0,\,t_1], \\[12pt]
    c, & t > t_1,
  \end{cases}
\]
where $q_n^{\mathrm{sep}}(\cdot\,;t_0)$ denotes the inspection probability under separation with boundary $t_0$, and $t_1$ is the unique solution to
\[
  (t_1 - pc)\,q^{\mathrm{pool}}(t_1) \;=\; \int_{t_0}^{t_1} q^{\mathrm{sep}}(x; t_0)\,dx.
\]
Moreover, let $\mu_{t_0}^*$ be a belief system that is consistent with $\sigma^*(\cdot; t_0)$ on path, and off path, satisfies the following: for any subsidy in $(0, t_0/p)$ it assigns belief $\delta_0$, and for any subsidy in $(\sigma^*(t_1; t_0), c)$ it assigns belief $\delta_{t_1}$. In other words, $\mu_0$ interprets an off-path low subsidy as coming from type 0, and at the discontinuity at the top it assigns belief to the cutoff type $t_1$.\footnote{We write $\delta_x$ for the degenerate distribution that places unit mass on type $x$.}
\end{definition}
 
For $n>2$, every element of this family exhibits the step--increasing--step pattern: low types offer no subsidy, intermediate types separate, and high types pool at the full subsidy.%
\footnote{When $n=1$, $q^{\mathrm{sep}}(t;t_0) = \1\{t \geq t_0\}$, so $\sigma^*(t;t_0)$ reduces to a step function: $\sigma^*(t;t_0) = (t_0/p)\,\1\{t \geq t_0\}$.}
Moreover, all members of this family are supported on equilibrium.
\begin{lemma}\label{lem:optimal-SIS-eqm}
    Fix $t_0 \in [\underline{t}, \max\{cp, 1\}]$. The assessment $\big(\sigma^*(\cdot; t_0), DSIP_{\mu_{t_0}^*}, \mu_{t_0}^*\big)$ is an equilibrium.
\end{lemma}
Hence the equilibrium set is large and nontrivial.\footnote{The SIS family does not exhaust the set of equilibria. In particular, there also exist equilibria with interior pooling. In \cref{subsec-app:property-eqm} we illustrate additional properties that equilibrium subsidy policies may satisfy.} %
Among this family of equilibria, we single out the the equilibrium $\big(\sigma^*(\cdot; \underline{t}), DSIP_{\mu_{\underline{t}}^*}, \mu_{\underline{t}}^*\big)$, which induces inspection of all types down to the lowest possible cutoff $\underline{t}$. We refer to this outcome as the \emph{reasonable equilibrium}. It is efficient because it maximizes the set of types that can be inspected: each inspected type offers a subsidy that is profitable for the firm and simultaneously optimal for the consumer to inspect. As a result, this equilibrium generates the most informative allocation of attention.

\subsection{Equilibrium Selection}\label{subsec:refinement}

We now apply a forward–induction refinement, following \citet{cho1987signaling, banks1987equilibrium}, to narrow the set of equilibria. Off–path beliefs must assign zero probability to types for whom a deviation is dominated. We adopt the equilibrium–dominance restriction. Our main result is that the unique equilibrium outcome surviving this refinement is precisely the one induced by the reasonable equilibrium.

\begin{theorem}\label{thm:uniqueness}
For $n \geq 2$, any equilibrium that fails to induce the outcome of $\bigl(\sigma^*(\cdot;\underline{t}), DSIP_{\mu^*_{\underline{t}}}, \mu^*_{\underline{t}}\bigr)$ fails the equilibrium–dominance refinement.
\end{theorem}

To establish the result, consider \emph{optimistic beliefs at $s$}: the belief system where, after an off–path subsidy $s$, the posterior assigns probability one to type-$1$, while every other off–path $s'\neq s$ is assigned probability one to type-$0$ and never inspected. On–path beliefs remain consistent with $\sigma$, so this constitutes an admissible continuation. To highlight the role of $s$ we denote the induced attention by $q_{\sigma, s}$. Under such beliefs, the inspection index reduces to \( r_{\mu'}(s)=u-(c-s)\).
Hence any $s>c-u$ is inspected with positive probability, although $q_{\sigma,s}(s)<1$ whenever there exists an on-path $s'$ with $s'>s$ and $(c-s')/(c-s)<\tau_{\mu}(s')$.%
\footnote{If $c > u$, then firms need to offer a high-enough subsidy to induce inspection. Thus, even under optimistic beliefs, some subsidies may never be inspected as the net search cost $c - s$ is still higher than the benefit $u$.}

\begin{lemma}\label{lem:EDR-characterization}
Fix an equilibrium $(\sigma,\iota,\mu)$. Assume there exists an interior pooling at $s'$, an off–path subsidy $s$ and a type $t$ with $\sigma(t) = s'$ such that, under optimistic beliefs at $s$, $t$ strictly prefers deviating to $s$. That is, 
\[
    (t-ps)\,q_{\sigma,s}(s)\;>\;\pi_{\sigma,\iota}(\sigma(t);t).
\]
Then, $(\sigma,\iota,\mu)$ fails the equilibrium-dominance refinement.
\end{lemma}

This lemma provides a simple test for equilibrium selection. We now sketch the proof of uniqueness. The strategy is to rule out, step by step, any equilibrium configuration that fails the equilibrium–dominance refinement under optimistic beliefs. First, we show that any off–path subsidy $s$ must receive at least as much attention as the lowest on–path subsidy above it would generate under full separation. Using this fact together with the test in \cref{lem:EDR-characterization}, we show that no equilibrium with interior pooling in $[0,c)$ that attracts positive attention can survive. Moreover, a discontinuity can survive only if it occurs at the cap, since any interior jump with attention on both sides would trigger a profitable deviation under optimistic beliefs. Thus the only equilibrium candidates that remain are SIS equilibria. We then establish that the lowest inspected type must coincide with $\underline{t}$, which delivers the uniqueness result.

Fix an equilibrium $(\sigma,\iota,\mu)$. Suppose there exists subsidy levels $s'<s$, where $s'$ is on–path with $q_{\sigma,\iota}(s')>0$ and $s$ is off–path. Under optimistic beliefs at $s$, the posterior puts maximal weight on match quality. Hence no lower subsidy, on–path or off–path, can be inspected before $s$. Since, by assumption, some types choose higher on–path subsidies, $s$ must receive at least the attention given to the lowest such type under full separation. Formally,
\[
q_{\sigma,s}(s)\;\ge\;\inf_{\{t:\,\sigma(t)>s\}}q^{\mathrm{sep}}(t;0).
\]
In this worst case, all such higher types also generate a higher inspection index and thus are inspected before.

Now assume that $(\sigma,\iota,\mu)$ has an interior pooling level $s\in(0,c)$ that is inspected with positive probability, i.e.\ $q_{\sigma,\iota}(s)>0$. Let $t^*$ be the highest type in the pool, with $t^*<1$. Then $\sigma$ cannot be right–continuous at $t^*$.\footnote{Formally, let $T_\sigma(s)=\{t\in T:\sigma(t)=s\}$, and say that there is pooling at $s$ if $T_\sigma(s)$ has positive $F$–measure. In any equilibrium with such pooling, $\sigma$ must be discontinuous at the top of the pool. See \cref{lem-app:no-subcap-pooling} for the general argument.} In other words, there exists $\varepsilon>0$ such that every $s'\in(s,s+\varepsilon)$ is off–path. And since $t^*<1$, there must also exist some higher on-path subsidy. Thus, using the previous logic, under optimistic beliefs for any $s'$, $q_{\sigma, s'}(s') \geq q^\textrm{sep}(t^*)$.
By \cref{thm:sorting}, inspection probability is increasing in the subsidy. Because there is a mass at $s$, types just above the pool receive strictly more attention: there exists $\eta>0$ such that
\[
\liminf_{\varepsilon\downarrow 0}\big[q_{\sigma,\iota}\big(\sigma(t^*+\varepsilon)\big)-q_{\sigma,\iota}(s)\big]\;\ge\;\eta.
\]
Moreover, the limit from the right can be identified with the separation benchmark. For any type, and in particular for those just above $t^*$, the worst attention case is that all higher types are inspected with probability one. In that case the inspection probability for such a type coincides with the full–separation benchmark $q^{\textrm{sep}}(\cdot;0)$. Since $q^{\textrm{sep}}$ is continuous, taking the limit as $\varepsilon\downarrow 0$ yields
\[
\liminf_{\varepsilon\downarrow 0} q_{\sigma,\iota}\big(\sigma(t^*+\varepsilon)\big)=q^{\textrm{sep}}(t^*).
\]
Therefore, under optimistic beliefs, any off-path $s'$ then secures attention at least as high as $q^{\textrm{sep}}(t^*)$ which is bounded above from the on-path attention $q_{\sigma,\iota}(s)$, and thus 
\[
q_{\sigma,s'}(s')-q_{\sigma,\iota}(s)\;\ge\;\tfrac\eta2 \;>\;0.
\]
Next, note that type $t^*$ must earn positive profit; otherwise, lower types in the pool would prefer to deviate to no subsidy. Hence $t^*-ps>0$. For any $s'$ close to $s$, comparison of payoffs yields
\[
\begin{split}
    (t^*-ps')\,q_{\sigma,s'}(s')-(t^*-ps)\,q_{\sigma,\iota}(s)
    &= (t^*-ps)\big[q_{\sigma,s'}(s')-q_{\sigma,\iota}(s)\big]-p(s'-s)\,q_{\sigma,s'}(s') \\
    &\ge\;(t^*-ps)\,\tfrac\eta2\;-\;p(s'-s),
\end{split}
\]
since $0<q_{\sigma,s'}(s')\le 1$. Choosing $s'$ sufficiently close to $s$ (for example, $s' = s+(t^*-ps)\eta/(4p)$) makes the payoff strictly higher. Thus $t^*$ has a profitable deviation under optimistic beliefs. We conclude, by \cref{lem:EDR-characterization}, that any equilibrium with a positively inspected pooling below $c$ fails the equilibrium–dominance refinement.

A similar argument rules out any equilibrium with a discontinuity inside an inspected region that does not reach the full subsidy. At such a point, the ``jumping” type has a profitable deviation by lowering its subsidy: this both raises its payoff and, under optimistic beliefs, leaves its inspection probability at least unchanged. Since both effects increase profit, the deviation is strictly profitable. The case of a discontinuity at the cap $s=c$ is different. Types pooling at the top are \emph{always} inspected first. If a type in this pool lowers its subsidy, it is inspected only after all top–pool types fail to match, even under optimistic beliefs. In this case, lowering the subsidy strictly reduces inspection probability and may lower profit.

Therefore any equilibrium with an interior pool or an interior discontinuity fails the equilibrium–dominance refinement. The only remaining candidates are the step-increasing-step policies (see \cref{def:optimal-SIS-family}). To conclude, we argue that if the lowest inspected type is $t_0>\underline{t}$, then under optimistic beliefs there is a profitable deviation. Pick $\varepsilon>0$ so that and $t':=t_0-\varepsilon>\underline{t}$. Under SIS, $t'$ offers no subsidy, and either it is part of a pooling region with positive inspection, in which our previous argument shows there is a profitable deviation, or it is part of a pool that is never inspected, i.e. $q_{\sigma,\iota}(0)=0$. We focus on this second case. Consider the off–path deviation $s':=\underline{t}/p$. Then
\[
t'-ps'\;=\;t'-\underline{t}\;>\;0,
\]
and by construction we also have $q_{\sigma,s'}(s')\geq q^\textrm{sep}(\underline{t}, 0) >0$. Thus $t'$ has a profitable deviation under optimistic beliefs, so the putative equilibrium fails the equilibrium–dominance refinement.

Together, these lemmas imply that any equilibrium surviving refinement must induce an outcome that coincides with the outcome of the reasonable SIS: no inspection below the participation cutoff, strict revelation on an interval, and a single pooling block at the cap. The bottom cutoff is pinned down by participation at $\underline{t}$, and the top cutoff by boundary indifference at $c$. By construction, this equilibrium is also the most informative: it has the largest separating region among all equilibria and induces participation from every firm type that can secure positive inspection. Any remaining differences appear only at the boundaries or within the no–inspection region.

\section{Welfare and Comparative Statics}\label{sec:welfare}

We now compute welfare under the reasonable equilibrium, decomposing it into consumer surplus, producer surplus, and total efficiency.  
We then establish comparative statics with respect to the primitives $(p,c,n)$.

The next lemma summarizes key properties of the lower and upper cutoff, the attention function, and the separating subsidy policy. 

\begin{lemma}\label{lem:SIS-properties}
The following properties hold:
\begin{enumerate}[label = (\roman*)]

    \item \textbf{Lower cutoff.}  
    The lower cutoff is given by $\underline{t}(p,c) = pc/(1+p)$. It is strictly increasing in both $c$ and $p$, and independent of $n$. 

    \item \textbf{Attention function. }  
    For any fixed $t_0$, the separating attention function $q_n^{\mathrm{sep}}(t; t_0)$ is strictly decreasing in $n$ for all $t \in (t_0,1)$. Moreover, for each $n$ and $t$, $q_n^{\mathrm{sep}}(t;\cdot)$ is decreasing. Therefore, $q_n^{\mathrm{sep}}(\cdot;\underline{t})$ is decreasing in $c$ and $p$.

    \item \textbf{Separating policy.}  
    For all $t \geq t_0'$, $\sigma_n^{\mathrm{sep}}(t; t_0') > \sigma_n^{\mathrm{sep}}(t; t_0)$ with $t_0' > t_0$.  
    In addition, $\sigma_{n+1}^{\mathrm{sep}}(t; t_0) > \sigma_n^{\mathrm{sep}}(t; t_0)$ for all $t > t_0$. We have that $\sigma_n^{\mathrm{sep}}(t; \underline{t})$ is increasing in $c$ and decreasing in $p$.

    \item \textbf{Upper cutoff.}  
    The upper cutoff $\bar t$ is strictly increasing in both $p$ and $n$.

\end{enumerate}
\end{lemma}

A higher inspection cost $c$ or a higher price of subsidies $p$ raises the entry cutoff $\underline{t}$, excluding more low–quality firms. As the number of firms $n$ grows, competition for attention intensifies and the inspection probability of each individual firm falls. To compensate, firms respond with higher subsidies, making the separating schedule steeper. Finally, the upper cutoff $\bar t$ increases with $p$ and $p$, since more expensive subsidies compress the set of types willing to separate before pooling at the cap.

\paragraph{Notation.} For convenience, let $q^*(t)$ denote the attention that a type $t$ receives in the reasonable equilibrium: it equals $q^{\mathrm{sep}}(t;\underline{t})$ for $t<\bar t$, and $q^{\mathrm{pool}}(\bar t)$ for $t\geq \bar t$.

\subsection{Welfare Accounting}\label{subsec:welfare-accounting}

Define the average inspection probability per firm as the expected \attention a firm generates:
\begin{equation}\label{eq:qbar}
Q_n(p,c)
\;\equiv\; \mathbb{E}_t[ q^*(t)] = 
\int_{\underline{t}}^{\bar t} q_n^{\mathrm{sep}}(t)\,dF(t)
\;+\;
q_n^{\mathrm{pool}}(\bar t)\,[1-F(\bar t)].
\end{equation}

A type-$t$ firm on the separating branch pays $p\,\sigma^{\mathrm{sep}}(t;\underline{t})$ upon inspection; envelope arguments yield for $t\in[\underline{t},\bar t)$
\[
p\,\sigma^{\mathrm{sep}}(t;\underline{t})\,q_n^{\mathrm{sep}}(t;\underline{t})
\;=\;
t\,q_n^{\mathrm{sep}}(t; \underline{t})
\;-\;
\int_{\underline{t}}^{t} q_n^{\mathrm{sep}}(x; \underline{t})\,dx.
\]
Pooled types $t\ge \bar t$ pay $p\,c$ upon inspection, so their expected payment equals $p\,c\,q_n^{\mathrm{pool}}(\bar t)$.
Thus, on expectation a firm expected transfer costs is given by 
\begin{equation*}\label{eq:Tn}
\begin{aligned}
\phi_n^{F}(p,c)
&:=\;\mathbb{E}_t\!\left[p\,\sigma^*(t)\,q^*(t)\right]\\
&=\;\int_{\underline{t}}^{\bar t}\!\Bigl(t\,q_n^{\mathrm{sep}}(t;\underline{t})-\int_{\underline{t}}^{t}q_n^{\mathrm{sep}}(x;\underline{t})\,dx\Bigr)\,dF(t)
\;+\;p\,c\,q_n^{\mathrm{pool}}(\bar t)\,[1-F(\bar t)].
\end{aligned}
\end{equation*}

No match occurs precisely when all inspected firms fail to match which happens with probability $\left(1 - \int_{\underline{t}}^1 x\,dF(x)\right)^{n}$.%
\footnote{The result follows the same logic we use to construct $q^\textrm{sep}$, where now we could think that the outside option is a firm that is inspected after all firms with type above $\underline{t}$ fail to match. }
From this observation, the probability of a match is
\begin{equation*}\label{eq:match-prob}
m_n(p,c)
=\; 1 - \left(1 - \int_{\underline{t}}^1 x\,dF(x)\right)^{n}.
\end{equation*}
The \emph{consumer} expected net inspection cost for a single firms equal
\begin{equation*}\label{eq:TnC}
C_n(p,c)
\;:=\;\mathbb{E}_t\!\left[(c-\sigma^*(t))\,q^*(t)\right]
\;=\; c\,Q_n(p,c)\;-\;\frac{1}{p}\phi_n^F(p,c).
\end{equation*}

\paragraph{Welfare decomposition.}
Write consumer surplus and producer surplus as
\begin{align}
CS_n(p,c; t_0)
&:=\; u\,m_n(p,c)\;-\;n\,C_n(p,c),\label{eq:CS}\\
PS_n(p,c; t_0)
&:=\; n\,\mathbb{E}_t\!\big[(t-p\,\sigma^*(t))\,q^*(t)\big]
\;=\;n\,\mathbb{E}_t[t\,q^*(t)]\;-\;n\,\phi_n^{F}(p,c),\label{eq:PS}
\end{align}
Social welfare adds up consumer and producer surplus:
\begin{equation}\label{eq:TotalWelfare}
\begin{split}
W_n(p,c) & 
\;:=\;CS_n(p,c; t_0)+PS_n(p,c; t_0)
\\ & = m_n(p,c; t_0) + n \big(\,\mathbb{E}_t[(t - c)\,q^*(t)]\big)- n \left(1 - \frac{1}{p}\right)\phi_n^{F}(p,c) 
\end{split}
\end{equation}

In the special case where $p = 1$ social welfare collapses and results into the expected match value minus the search costs incurred into achieving such match. 
\subsection{Comparative Statics}\label{subsec:comparative-statics}
We record the main monotonicity properties of the allocation and welfare components. 

\begin{theorem}[Comparative statics in the reasonable equilibrium]\label{thm:comparative-statics-SIS}
Then:
\begin{enumerate}[label=(\roman*)]

\item \textbf{Search intensity.} $Q_n(p,c)$ is decreasing in $p$ and in $c$, and  increasing in $n$.

\item \textbf{Match probability.} $m_n(p,c)$ is decreasing in $p$ and in $c$, and  increasing in $n$.

\item \textbf{Consumer surplus.} $CS_n(p,c)$ is decreasing in $p$ and in $c$, and increasing in $n$.

\item \textbf{Producer surplus.} $PS_n(p,c)$ is decreasing in $p$ and $c$. Moreover, the per-firm surplus is decreasing in $n$, i.e.  $PS_n(p,c)/n$ decreases in $n$.
\end{enumerate}
\end{theorem}





\section{Discussion: Agentic Search}\label{sec:agentic}

The rise of the \emph{agentic economy} is transforming how search operates \citep{Rothschildetal2025}. Consumers increasingly delegate discovery and evaluation to AI agents that inspect products on their behalf. In this setting, inspection becomes a sequence of auditable computational events (e.g., API calls, token usage, or verification steps) that carry a measurable cost. Attention thus becomes a contractible input that can be priced, transferred, and subsidized.

We embed this development within our framework by modeling a platform that hosts firms and sells \emph{inspection tokens}.\footnote{A \emph{token} is the basic unit of text processed by a language model, typically a word, part of a word, or punctuation mark, converted into numerical representations the model can manipulate mathematically.} Each inspection consumes one token at a real cost $c > 0$ to the consumer. Firms can purchase tokens from the platform at price $p \geq 0$ and transfer them to the agent, effectively subsidizing inspection costs. This tokenized approach provides a parsimonious foundation for analyzing efficiency and surplus in emerging agentic marketplaces.

\subsection{Platform Revenue Optimization}

The platform's revenue problem is setting a token price to maximize revenue. Revenue equals the token price multiplied by the expected token demand, where demand precisely captures the volume of subsidies firms provide. For each price $p \geq 0$, let $(\sigma^*(\cdot;p), \text{DSIP}, \mu^*)$ denote the unique refined outcome of $\mathcal{G}(p)$ under equilibrium-dominance (see \cref{sec:consumer-optimal-eqm}, \cref{thm:uniqueness}). Let $q^*(t;p)$ be the inspection probability induced by DSIP. We define the per-firm \emph{subsidy demand} as:
\[
\begin{split}
    D(p) & := \mathbb{E}\!\left[\sigma^*(t;p) \, q^*(t;p)\right]
\\ & = \int_{\underline{t}(p)}^{\bar{t}(p)} \sigma^{\mathrm{sep}}(t) \, q^{\mathrm{sep}}(t;\underline{t}(p)) \, dF(t) 
+ c \, q^{\mathrm{pool}}(\bar{t}(p))[1-F(\bar{t}(p))],
\end{split}
\]
where $\underline{t}(p) = \frac{pc}{1+p}$ is the participation cutoff and $\bar{t}(p)$ the pooling cutoff. Platform revenue is then:
\[
R(p) = n \, p \, D(p).
\]

Let $\psi(t) = t - \frac{1-F(t)}{f(t)}$ be the virtual value function. Per-firm revenue admits the decomposition:
\begin{equation}\label{eq:rev-decomp}
\frac{R(p)}{n}
=
\underbrace{\left(\int_{\underline{t}(p)}^{\bar{t}(p)} \psi(x) \, q^{\mathrm{sep}}(x;\underline{t}(p)) \, \frac{dF(x)}{F(\bar t(p))}\right)}_{\text{separating branch}} F(\bar t(p)) 
+
\underbrace{\Big(\bar{t}(p) \, q^{\mathrm{pool}}(\bar{t}(p))\Big)}_{\text{pooling branch}} \, [1-F(\bar{t}(p))].
\end{equation}

This decomposition reveals two distinct revenue components.\footnote{See \cref{lem-app:useful-identities} in the appendix to see how we reach such representation.} The first term represents the \emph{revealing branch}, where firms separate by type and revenue corresponds to expected virtual surplus weighted by inspection probabilities. This component exhibits a single-peaked relationship with price: lower prices expand subsidies and increase inspection, but beyond the unit-elasticity point, further price cuts reduce revenue. At high prices, the pooling region vanishes and only revelation contributes to revenue.

The second term captures the \emph{pooling branch}, where all types above $\bar{t}(p)$ are grouped and effectively pay the cutoff virtual value $\bar{t}(p)$. This component increases monotonically in price: higher prices raise per-unit revenue from the pooled types, though the pool's relative weight shrinks as participation decreases.

Crucially, price does not merely scale revenue, it selects the \emph{margin of competition}. Lower prices reduce per-unit revenue but broaden participation,  and increase subsidies, further it expands the pooling region, giving greater weight to the pooling branch. Higher prices shrink participation and thin the pool, though each inspected type provides larger subsidies. In this regime, the pooling branch matters less and revenue is primarily driven by  the separating branch. Overall, revenue combines a hump-shaped revealing branch with a monotone pooling branch. The unique revenue-maximizing price $p^*$ emerges where these opposing forces balance.

\begin{theorem}[Excess Search]\label{thm:excess-search}
Suppose $F$ satisfies the regularity condition. Then the platform's revenue maximization problem admits a unique solution $p^* > 0$. At this optimal price:
\begin{enumerate}
\item \textbf{Pooling remains active:} $\bar{t}(p^*) < 1$.
\item \textbf{Negative virtual values receive inspection:} Some types with $\psi(t) < 0$ are inspected with positive probability.
\end{enumerate}
\end{theorem}
 
The platform optimally sets $p^*$ to thicken the top pool, steering attention toward inspections that a virtual surplus benchmark would exclude. Since generating additional attention is cheapest in this pooling region, profit maximization occurs precisely when the market is induced to search excessively. This creates a fundamental tension in agentic marketplaces. While consumers benefit from increased match likelihood and reduced net inspection costs, attention becomes distorted toward excessive inspection, obfuscating discovery. The platform's revenue maximization thus generates a welfare loss: the market searches too much relative to the first-best allocation.


\section{Concluding Remarks}\label{sec:conclusion}

This paper develops a theory of attention markets in which firms can directly subsidize costly search. In such environments, subsidies play a dual role. They serve as signals, conveying information about a firm’s expected match value. At the same time, they are instrumental, reducing the consumer’s effective inspection costs and thereby altering the sequence of search. In equilibrium, these two roles jointly determine which firms participate, which reveal their type, and which pool. In this way, subsidies endogenously structure the allocation of attention.

We show that every equilibrium satisfies a \emph{subsidy–sorting principle}: higher types weakly offer higher subsidies, and the consumer follows a descending–subsidy inspection order. Within this structure, the unique equilibrium outcome surviving refinement is the \emph{reasonable equilibrium}. It takes a step–increasing–step form: no inspection below a participation cutoff, revelation on an intermediate region, and a pooling block at the cap. This equilibrium is efficient in the sense of being most informative: it maximizes the set of types that can induce inspection, subject to incentive constraints. Any information that is lost is either uninspected altogether or free to obtain.

These results provide a benchmark for platform design. If a platform can mint and sell inspection tokens, its problem reduces to a linear–cost mechanism. Yet profit motives may distort attention. By thickening the top pool, a platform can steer consumers to inspect too much, raising revenue but lowering social efficiency. 

The framework also delivers empirical predictions. Subsidies should be weakly increasing in observable proxies for match quality. High–subsidy pools should exhibit congestion, with inspection probabilities falling as the pool thickens. Changes in token prices should shift the participation margin and alter match rates and transfers in predictable ways. Such patterns can be measured in platform data, offering a path for empirical validation.

More broadly, the analysis highlights how inspection can be organized without auctions. Direct subsidies decentralize the allocation of attention, sustaining efficiency so long as ranking rules remain transparent. As AI agents increasingly intermediate search, this principle provides both a microfoundation for programmable attention and a policy guide for preserving efficiency in environments where the temptation to steer is strongest.

\begin{singlespace}
{\small
	\addcontentsline{toc}{section}{References}
	\setlength{\bibsep}{.25\baselineskip}
	\bibliographystyle{aer}
	\bibliography{references}
}
\end{singlespace}

\newpage
\appendix

\section{Omitted proofs for \cref{sec:sorting}}\label{appendix:sorting-proofs}

\subsection{Proof of \cref{lem:subsidy-policy-weakly-increasing}}

First, we establish that the attention function is weakly increasing on path.

\begin{claim}
    For $s,s' \in \textrm{Im}(\sigma)$, if $s' \geq s$, then $q_{\sigma, \iota}(s') \geq q_{\sigma, \iota}(s)$.
\end{claim} 

\begin{proof}
    Assume towards a contradiction that this is not the case, then $q_{\sigma, \iota}(s') < q_{\sigma, \iota}(s)$. Let $t'$ be such that $\sigma(t') = s'$, then from the firm optimality it should hold that
    \[ (t' - ps')q_{\sigma, \iota}(s') \geq (t - ps) q_{\sigma, \iota}(s) \]
    and in particular, for $s$. However, note that 
    \[ (t' - ps')q_{\sigma, \iota}(s') < (t - p\sigma(t))q_{\sigma, \iota}(s) < (t' - ps)q_{\sigma, \iota}(s). \]
    Thus, type $t'$ has a profitable deviation. 
\end{proof}

Note the proof establish something stronger. In particular, we have that if $s \in \textrm{Im}(\sigma)$, then for all $s' < s$, $q_{\sigma, \iota}(s') \leq q_{\sigma, \iota}(s)$ regardless of $s'$ being on-path or off-path.

Second, we show that the firm payoff over inspection probabilities satisfies single-crossing differences in type. Thus, higher types firms are willing to pay more for the same level of attention. In turn, this means that the subsidy policy should be increasing.

\begin{claim}
    The subsidy policy is increasing: if $t' \geq t$, then $\sigma(t') \geq \sigma(t)$.
\end{claim}
\begin{proof}
Let \(q(s)\) denote the ex-ante probability that a firm offering subsidy \(s\) is inspected in equilibrium. Define \(q(t) := q(\sigma(t))\) for shorthand. Since \(\sigma(t)\) is an equilibrium strategy, it must solve:
    \[
        \sigma(t) \in \arg\max_{s \in [0,c]} \big[t - ps\big] \cdot q(s),
    \]
which implies \(\big[t - p\sigma(t)\big] \geq 0\) since \(s = 0\) is always feasible and \(t, q(\cdot) \in [0,1]\)
    
For any two types \(t' > t\), it holds as follows:
    \begin{align}
        \big[t - p\sigma(t)\big] \cdot q(t) &\geq \big[t - p\sigma(t')\big] \cdot q(t'), \label{eq:t-opt} \\
        \big[t' - p\sigma(t')\big] \cdot q(t') &\geq \big[t' - p\sigma(t)\big] \cdot q(t). \label{eq:tprime-opt}
    \end{align}
     Re-arrenging \eqref{eq:t-opt} and \eqref{eq:tprime-opt}, we obtain:
    \begin{align*}
        p \cdot \big[\sigma(t') q(t') - \sigma(t) q(t)\big] &\geq t \cdot \big[q(t') - q(t)\big], \\
        t' \cdot \big[q(t') - q(t)\big] &\geq p \cdot \big[\sigma(t') q(t') - \sigma(t) q(t)\big].
    \end{align*}
    Combining the two inequalities implies
\[
    t'\big[q(t') - q(t)\big] \;\;\geq\;\; t\big[q(t') - q(t)\big],
\]
or equivalently,
\[
    (t' - t)\cdot\big[q(t') - q(t)\big] \;\;\geq\;\; 0.
\]
Because $t' > t$, it follows that $q(t') \geq q(t)$. That is, inspection probabilities are weakly increasing in type.  

Finally, since the inspection function $q(\cdot)$ is weakly increasing in subsidies (by the previous claim), it must hold that $\sigma(t') \geq \sigma(t)$. In other words, equilibrium subsidies are monotone in type. 
\end{proof}

\subsection{Proof of \cref{lem:descending-search}}

\begin{proof}
By Weitzman's reservation-value theorem, the consumer's optimal inspection policy is given by the index rule: inspect firms in decreasing order of their reservation (index) values. If several firms share the same index, any symmetric tie-breaking rule yields the same expected payoff; under the symmetry assumption the only way to break ties is to randomize uniformly among tied firms. Hence the index rule with uniform tie-breaking (when needed) is optimal.
\end{proof}

\subsection{Proof of \cref{cor:sorting-on-path}}

\begin{proof}
    On-path higher subsidies correspond to higher types, moreover higher subsidies also decrease the search cost. Both forces imply that search a higher subsidy is associeted with a higher reservation index. Thus, on-path the consumer searches in descending order of subsidy.  The optimal stopping point is determined when $r_\mu(s) < 0$. That is,
    \[ u - \frac{c - s}{\tau_\mu(s)} < 0 \iff u \mathbb{E}[t \mid \sigma(t) = s] < c-s \]
\end{proof}

\subsection{Proof of \cref{thm:sorting}}

\begin{proof}
    The result immediately follows from \cref{lem:subsidy-policy-weakly-increasing} and \cref{lem:descending-search}. In its current statement, we used the result in \cref{cor:sorting-on-path}. 
\end{proof}

\section{Technical details and supplementary analysis for \cref{sec:consumer-optimal-eqm}} \label{appendix:consumer-optimal-proofs}

\subsection{Inspection probability under full separation} \label{app:inspection-prob}

\begin{lemma}[Inspection probability under full separation] \label{lemma:inspection-prob}
Suppose the subsidy policy $\sigma$ is strictly increasing and the consumer follows the descending–subsidy index rule $DSIP_\mu$ where $\mu$ is on-path consistent with $\sigma$. Let $\underline{t}\in[0,c)$ denote the lowest inspected type, defined by $\underline{t} = \inf\{t:\sigma(t)=c-ut\}$. Then, for any $t\in T$, the ex-ante inspection probability of a type-$t$ firm is equivalent to
\[
q_n^{\text{sep}}(t;\underline{t})
   \;=\;
   \Biggl(1-\int_t^1 x\,dF(x)\Biggr)^{n-1}\mathbf{1}\{t\ge \underline{t}\}.
\]
The attention function $q_n^{\text{sep}}(\cdot;\underline{t})$ is strictly increasing on $(\underline{t},1)$, hence differentiable a.e., and $q_n^{\text{sep}}(1;\underline{t})=1$.
\end{lemma}

\begin{proof}

First, under $\sigma$ strictly increasing, and the inspection policy following $DSIP_\mu$, the consumer inspect all firms that offer a subsidy level $s \geq \sigma(\underline{t})$. 

For $x \geq \underline{t}$, consider a type $x$-firm. If such firm offers the $k$-th highest subsidy (rank $k$), its probability of inspection is given by expected failure rate $\prod_{\ell=1}^{k-1}\!\bigl(1-t_{(\ell:n-1)}\bigr)$ when $t_{(\ell:n-1)} \geq x$, where $t_{(\ell:n-1)}$ is the $\ell$-th highest type among the other $n-1$ firms. Hence the ex-ante inspection probability is given by 
\[
    q_{\sigma, DSIP_\mu}(x)
    \;=\;
    \sum_{k=1}^{n}
    \E\!\Biggl[
        \prod_{\ell=1}^{k-1}\!\bigl(1-t_{(\ell:n-1)}\bigr)
        \,\Bigm|\, t_{(k-1:n-1)} \ge x > t_{(k:n-1)}
    \Biggr]
    \Pr\!\bigl\{t_{(k-1:n-1)} \ge x > t_{(k:n-1)}\bigr\},
\]
with the convention $t_{(0:m)}:=1$, and $q_{\sigma, DSIP_\mu}(x) = 0$ for $x<\underline{t}$.

Fix $x\in T$. For $k=1,\dots,N$, consider the event
\[
\{t_{(k-1:n-1)} \geq x > t_{(k:n-1)}\},
\]
which is exactly the event that $k-1$ out of the $N-1$ opponent types lie above $x$.  
Since each opponent exceeds $x$ with probability $1-F(x)$, the count $K$ of such types follows a binomial law: $K \sim \mathrm{Bin}(N-1,\,1-F(x))$.  
Thus
\[
\Pr\{t_{(k-1:n-1)} \geq x > t_{(k:n-1)}\}
= \Pr\{K = k-1\}
= \binom{N-1}{k-1} (1-F(x))^{k-1}F(x)^{N-k}.
\]

Conditional on this event, the $k-1$ opponents above $x$ are i.i.d.\ draws from the truncated distribution $F(\cdot\mid t \geq x)$.  
Hence,
\[
\E\!\left[\prod_{\ell=1}^{k-1} (1-t_{(\ell:n-1)}) \;\middle|\; t_{(k-1:n-1)} \geq x > t_{(k:n-1)}\right]
= \Bigl(\E[1-t \mid t\geq x]\Bigr)^{k-1}.
\]

Therefore,
\[
\begin{aligned}
q_{\sigma, DSIP_\sigma}(x)
&= \sum_{k=1}^n
\left(\E[1-t \mid t\geq x]\right)^{k-1}
\binom{n-1}{k-1}(1-F(x))^{k-1}F(x)^{n-k} \\[6pt]
&= \sum_{k=0}^{n-1} \binom{n-1}{k}
\Bigl((1-F(x))\!\left(\E[1-t \mid t\geq x]\right)\Bigr)^k
F(x)^{n-1-k}.
\end{aligned}
\]

The binomial theorem now gives
\[
q_{\sigma, DSIP_\mu}(x) = \left[F(x) + (1-F(x)) - \int_x^1 t\,dF(t)\right]^{n-1}
= \left(1 - \int_x^1 t\,dF(t)\right)^{n-1} = q_n^{\text{sep}}(x; \underline{t}),
\]
where we use $\E[1-t \mid t\geq x] = 1 - \frac{1}{1-F(x)}\int_x^1 t\,dF(t)$. Establishing the result. The remaining properties follow directly from inspection. 
\end{proof}

\subsection{Inspection probability under pooling at the top} \label{app:pooling}

\begin{lemma}\label{lem:inspection-probability-pooling}
If all types in $[t_1,1]$ pool at $c$, then the ex–ante inspection probability for a pooled firm is
\[
q^{\mathrm{pool}}_n(t_1) 
= \mathbb{E}_{K}\!\left[\frac{1-(1-\mathbb{E}[t\mid t\ge t_1])^{K+1}}{(K+1)\,\mathbb{E}[t\mid t\ge t_1]}\right],
\]
where $K\sim\mathrm{Bin}(n-1,1-F(t_1))$. Moreover, $q^\textrm{pool}$ is continuous and strictly increasing. 
\end{lemma}

\begin{proof}
    Let \(K\) be the number of \emph{other} firms that also lie in \([t_11]\); by i.i.d.\ types, \(K\sim\mathrm{Bin}(n-1,1-F(t_1))\).
    Conditional on \(K=k\), the pool size is \(m=k+1\) (including the focal firm).
    Uniform tie–breaking makes the inspection order among the \(m\) pooled firms uniformly random.
    Fix the focal firm and suppose it is the \(j\)-th in that random order.
    The consumer reaches it if, and only if, all \(j-1\) predecessors fail to match.
    Conditional on being in the pool, types are i.i.d.\ with the conditional law \(F(\cdot\mid t\ge t_1)\); hence by independence,
    \[
    \Pr\{\text{all }j-1\text{ predecessors fail}\mid t\ge t_1\}
    = \Bigl(\E[1-t\mid t\ge t_1]\Bigr)^{j-1}
    = (1-\tau(t_1))^{\,j-1}.
    \]
    Averaging over the uniformly random position \(j\in\{1,\dots,m\}\),
    \[
    \Pr\{\text{focal is inspected}\mid K=k\}
    = \frac{1}{m}\sum_{j=1}^{m}(1-\tau(t_1))^{\,j-1}
    = \frac{1-(1-\tau(t_1))^{m}}{m\,\tau(t_1)}.
    \]
    Finally, take expectations over \(K\) to obtain the stated \(q^\textrm{pool}(t_1)\). Continuity follows from $F$ having full support, and strict increasingness from $\mathbb{E}[t \mid t\geq t_1]$ being increasing in $t_1$.
\end{proof}

\subsection{Omitted proofs for \cref{sec:consumer-optimal-eqm}}

\begin{lemma}\label{lem:existence-cutoff}
Fix $t_0 \in T_0(p,c)$.  
Define the upper cutoff $t_1:T_0(p,c)\to[pc,1]$ as follows:  
if $\sigma^*(1;t_0)<c$, set $t_1=1$; otherwise, let
\[
H(x;t_0) \;=\; q^{\mathrm{pool}}(x)\big(x-pc\big) \;-\; 
q^{\mathrm{sep}}(x;t_0)(x - p \sigma^\textrm{sep}(x; t_0)),
\]
and define $t_1$ as the unique solution to $H(t_1;t_0)=0$.  
The mapping $t_0 \mapsto t_1(t_0)$ is well defined, continuous, and differentiable on $T_0(p,c)$.
\end{lemma}

\begin{proof}
We begin with boundary values of $H(\cdot;t_0)$.  
At $x=pc$,
\[
H(pc;t_0) \;=\; -\,q^{\mathrm{sep}}(pc;t_0)\,\bigl(pc - p\sigma^{\mathrm{sep}}(pc;t_0)\bigr)
= - \int_{t_0}^{pc} q^{\mathrm{sep}}(y;t_0)\,dy \;<\; 0.
\]
At $x=t_c$ (with $\sigma^\textrm{sep}(t_c; t_0) = c$),
\[
H(t_c;t_0) \;=\; (t_c-pc)\bigl(q^{\mathrm{pool}}(t_c) - q^{\mathrm{sep}}(t_c;t_0)\bigr) \;>\; 0,
\]
as for all $x$, $q^{\mathrm{pool}}(x) > q^{\mathrm{sep}}(x;t_0)$ (which also implies that $H(x; t_0) > 0$ for all $x > t_c$).  
Since $H(\cdot;t_0)$ is continuous, the intermediate value theorem guarantees at least one root $t_1 \in (pc,t_c)$.

\medskip
\noindent\emph{Uniqueness.}  
To establish uniqueness, we show that $H(\cdot;t_0)$ has a single crossing.  
Recall the envelope identity
\[
q^{\mathrm{sep}}(x;t_0)\,\bigl(x-p\sigma^{\mathrm{sep}}(x;t_0)\bigr)
= \int_{t_0}^x q^{\mathrm{sep}}(y;t_0)\,dy.
\]
It follows that
\[
H(x;t_0) \;>\;0 
\quad\Longleftrightarrow\quad
q^{\mathrm{pool}}(x) \;>\; 
\frac{\int_{t_0}^x q^{\mathrm{sep}}(y;t_0)\,dy}{x-pc} =: H_-(x).
\]
Differentiating yields
\[
H_-'(x)
= \frac{p\, q^{\mathrm{sep}}(x;t_0)}{(x-pc)^2}\,\bigl(\sigma^{\mathrm{sep}}(x;t_0)-c\bigr).
\]
Then $H_-'(t_c)=0$, $H_-'(x)<0$ for $x<t_c$, and $H_-'(x)>0$ for $x>t_c$.   
Since $q^{\mathrm{pool}}(x)$ is strictly increasing, then $q^{\mathrm{pool}}(x) - H_-(x)$ is strictly decreasing. Since we already establish the existance of a zero, this guarantees uniqueness.

Hence $H(\cdot;t_0)$ has a unique root $t_1$, which establishes existence and uniqueness of the cutoff. Moreover, such root is to the left of $t_c$. The implicit function theorem establishes that that $t_1(t_0)$ is differentiable, where 
\[ t_1'(t_0) = - \frac{q^\textrm{sep}(t_0; t_0)(t_1 - pc)}{(q^{\mathrm{pool}})'(t_1)(t_1 - pc)^2 + pq^\textrm{sep}(t_1)(c - \sigma^\textrm{sep}(t_1; t_0) )} < 0,\]
where we use that $pc<t_1 < t_c$, and thus $c > \sigma^\textrm{sep}(t_1; t_0)$.
\end{proof}

\begin{proof}[Proof of \cref{lem:optimal-SIS-eqm}]
Fix $t_0\in T_0(p,c)$ and consider the assessment $(\sigma^*(\cdot;t_0),\mathrm{DSIP}_{\mu^*_{t_0}},\mu^*_{t_0})$. By \cref{lem:descending-search}, for any belief system $\mu$ the descending–subsidy index policy $\mathrm{DSIP}_\mu$ maximizes the consumer’s expected utility. By construction, $\mu^*_{t_0}$ coincides with Bayes’ rule on the realized support of $\sigma^*(\cdot;t_0)$. Hence the consumer’s strategy is optimal and beliefs are on–path consistent. Types $t\le t_0$ pool at $s=0$. Their equilibrium payoff is
\[
\pi^*(0;t)=\bigl(t-p\cdot 0\bigr)\,q_{\sigma^*,\mathrm{DSIP}}(0)\;\ge\;0,
\]
since $q_{\sigma^*,\mathrm{DSIP}}(0)\ge 0$. Any deviation that raises attention must jump to at least the bottom of the revealing branch, $
s\;\ge\;\sigma^*(t_0;t_0)\;=\;\frac{t_0}{p}$, because lower subsidies do not improve rank under $\mathrm{DSIP}_{\mu^*_{t_0}}$. For $t\le t_0$ this yields
\[
\pi(s;t)\;=\;\bigl(t-p s\bigr)\,q_{\sigma^*,\mathrm{DSIP}}(s)
\;\le\;\bigl(t-t_0\bigr)\,q_{\sigma^*,\mathrm{DSIP}}(s)
\;\le\;0,
\]
with strict inequality if $s>\tfrac{t_0}{p}$. Thus no $t\le t_0$ gains by deviating. For $t\in(t_0,\bar t)$ the policy $\sigma^*(\cdot;t_0)$ is strictly increasing and solves the standard envelope (local IC) condition in \eqref{eq:ODE-raw}, so truth telling maximizes interim profit along the branch. A deviation to a lower subsidy reduces rank and attention under $\mathrm{DSIP}_{\mu^*_{t_0}}$; a deviation to a higher subsidy cannot improve attention at that rank but strictly raises the per–inspection payment $ps$. Either way the deviation weakly lowers profit, with strict loss whenever rank changes or $s$ increases.

Combining, the consumer is best responding given $\mu^*_{t_0}$, beliefs are on–path Bayes consistent, and no firm type has a profitable deviation. Hence $(\sigma^*(\cdot;t_0),\mathrm{DSIP}_{\mu^*_{t_0}},\mu^*_{t_0})$ is a symmetric PBE.
\end{proof}

\section{Equilibrium Properties and the Refinement}\label{app:eqm-refinement}

This appendix collects general properties of equilibria and clarifies how the equilibrium-dominance refinement operates. These results underpin the arguments in \cref{subsec:refinement}.

By \cref{thm:sorting}, in any equilibrium $(\sigma,\iota,\mu)$ the subsidy policy $\sigma$ is weakly increasing and the consumer follows DSIP$_\mu$ (\cref{def:descending-subsidy-rule-policy}). Let $q_\sigma(s)$ denote the induced attention function for subsidy $s$.

\paragraph{Notation.}
Let $S=[0,c]$ be the set of feasible subsidies, $S_\sigma=\mathrm{Im}(\sigma)$ the on–path set, and $S_\sigma^-=S\setminus S_\sigma$ the off–path set. For $s\in S_\sigma$, let $T_\sigma(s)=\{t\in T:\sigma(t)=s\}$ be the set of types mapped to $s$. Let $q^{\mathrm{sep}}(t)=(1-\int_t^1 x\,dF(x))^{\,n-1}$ be the inspection probability under a strictly separating branch (derived in \cref{app:inspection-prob}), and let $q^{\mathrm{pool}}(\bar t)$ be the inspection probability of a pooled firm when all types in $[\bar t,1]$ pool at $c$ (derived in \cref{app:pooling}).

\subsection{General Properties of Equilibria}\label{subsec-app:property-eqm}

We first establish several properties that hold in any equilibrium.  
\begin{enumerate}[label=(\roman*)]
\item In any inspected region, at most one type earns zero profit.  
\item If there is interior pooling, the subsidy policy fails right–continuity.  
\item On a strictly separating branch, types must satisfy the differential equation derived earlier, with the appropriate initial condition.  
\end{enumerate}

Among types that receive attention, at most one can break even. Higher types must earn strictly positive rents.

\begin{lemma}\label{lem-app:single-zero-profit}
Fix an equilibrium $(\sigma,\iota,\mu)$.  
If there exist types $t_1<t_2$ with $q_{\sigma}(\sigma(t_1)),q_{\sigma}(\sigma(t_2))>0$, then $\pi_{\sigma,\iota}(\sigma(t_2);t_2)>0$.
\end{lemma}

\begin{proof}
Fix two types $t_1<t_2$ that both receive positive attention. By \cref{thm:sorting}, $q(\cdot)$ is weakly increasing in $S_\sigma$, and firms face payoffs $\pi(s;t)=(t-ps)q(s)$. If $\pi(\sigma(t_1), t_1) = (t - p\sigma(t_1)) q_\sigma(\sigma(t_1)) = 0$, then $t_1 = p\sigma(t_1)$. Since $t_2>t_1$, it follows that $\pi(\sigma(t_1); t_2)>0$, and by optimality $\pi(\sigma(t_2), t_2)\geq \pi(\sigma(t_1); t_2)>0$. Thus at most one inspected type breaks even; all others either earn rents or are not inspected.
\end{proof}

The next lemma shows that a  pool inside $(0,c)$ forces a discontinuity: the attention probability jumps upward just to the right of $s$, so a highest pooled type would deviate.

\begin{lemma}\label{lem-app:subcap-pool-disc}
Fix an equilibrium $(\sigma,\iota,\mu)$.  
If there exists $s\in(0,c)$ such that $\Pr\{t\in T_\sigma(s)\}>0$, then there exists $\varepsilon>0$ with $(s,s+\varepsilon)\subseteq S_\sigma^-$.
\end{lemma}

\begin{proof}
Suppose, toward a contradiction, that $[s,s+\varepsilon)\subseteq S_\sigma$. We show $q_\sigma$ has a strict right–hand jump at $s$.

Without loss of generality assume $\sigma^{-1}$ is strictly increasing on $(s,s+\varepsilon)$ (otherwise restrict the interval).  
For any $s'\in(s,s+\varepsilon)$, a firm at $s'$ is ranked strictly above those at $s$ and is inspected with probability \( q_\sigma(s') \;=\; q^{\mathrm{sep}}\!\big(\sigma^{-1}(s')\big)\).
A firm at $s$ faces pooling. Let $K$ be the number of other pooled firms at $s$, with $\Pr\{K\ge1\}>0$ since $\Pr\{t\in T_\sigma(s)\}>0$. Conditional on reaching level $s$ and on $K=k\ge1$, its inspection probability is
\[
q_{\mathrm{pool}}(k)\;=\;\frac{1-(1-\tau)^{k+1}}{(k+1)\tau}\;<\;1,
\]
where $\tau$ is the common pooled posterior success probability. Hence
\[
q_\sigma(s) \;=\; q^{\mathrm{sep}}\!\big(\sigma^{-1}(s)\big)\,\mathbb{E}\big[q_{\mathrm{pool}}(K)\big] \; < \; q^{\mathrm{sep}}\!\big(\sigma^{-1}(s)\big).
\]
Taking $s'\downarrow s$ and using continuity of $q^{\mathrm{sep}}\!\circ\,\sigma^{-1}$ on-path,
\[
\lim_{s'\downarrow s}\!\big(q_\sigma(s')-q_\sigma(s)\big)
\;=\;
q^{\mathrm{sep}}\!\big(\sigma^{-1}(s)\big)\,\Big(1-\mathbb{E}[q_{\mathrm{pool}}(K)]\Big)
\;>\;0,
\]
since $\E[q_{\mathrm{pool}}(K)]<1$ and $q^{\mathrm{sep}}(\sigma^{-1}(s))>0$.

Let $t^*=\sup T_\sigma(s)$; if $t^*\notin T_\sigma(s)$ pick a sequence $(t^k)$ in $T_\sigma(s)$ approaching $t^*$. Then for $s'$ close to $s$,
\[ \begin{split}
    \pi(t^*;s')-\pi(t^*;s) & = (t^*-ps')q_\sigma(s')-(t^*-ps)q_\sigma(s)
    \\& \;\ge\; (t^*-ps)\big(q_\sigma(s')-q_\sigma(s)\big)-p(s'-s),
\end{split} \]
Hence, as 
\[ \lim_{s' \rightarrow s} \Big(\pi(t^*;s')-\pi(t^*;s)\Big) = (t^*-ps)q^{\mathrm{sep}}\!\big(\sigma^{-1}(s)\big)\,\Big(1-\mathbb{E}[q_{\mathrm{pool}}(K)]\Big) > 0 \]
a contradiction. Hence $(s,s+\varepsilon)\subseteq S_\sigma^-$ for some $\varepsilon>0$.
\end{proof}

Finally, we show that on any separating branch the ODE at \eqref{eq:ODE-raw} precisely pins down the subsidy policy.

\begin{lemma}\label{lem-app:revealing-ODE}
Let $(\sigma,\iota,\mu)$ be an equilibrium. Suppose there exists an interval $(t_1,t_2)$ on which $\sigma$ is strictly increasing and $q_\sigma(\sigma(t_1))>0$. Then on $(t_1,t_2)$ $\sigma$ solves \eqref{eq:ODE-raw} on $(t_1,t_2)$, that is
\[
\sigma'(t)
\;=\;
\frac{\big(t-p\,\sigma(t)\big)}{p}\frac{q^{\mathrm{sep}\,\prime}(t)}{q^{\mathrm{sep}}(t)}.
\]
\end{lemma}

\begin{proof}
Fix $t\in(t_1,t_2)$. Since $\sigma$ is strictly increasing on $(t_1,t_2)$, subsidies identify types on this interval, and the induced attention along the branch satisfies for $x \in (t_1, t_2)$, \( q_\sigma\big(\sigma(x)\big)=q^{\mathrm{sep}}(x)\).
At an interior optimum $s=\sigma(t)$ the first-order condition holds:
\[
0 \;=\; \frac{d}{ds}\pi(s;t)\Big|_{s=\sigma(t)}
\;=\; -p\,q_\sigma\big(\sigma(t)\big)\;+\;(t-p\,\sigma(t))\,\frac{d}{ds}q_\sigma(s)\Big|_{s=\sigma(t)}.
\]
Along the revealing branch $q_\sigma(s)=q^{\mathrm{sep}}\!\big(\sigma^{-1}(s)\big)$, so by the chain rule,
\[
\frac{d}{ds}q_\sigma(s)\Big|_{s=\sigma(t)}
=\frac{q^{\mathrm{sep}\,\prime}(t)}{\sigma'(t)}.
\]
Substituting into the FOC, using $q_\sigma(\sigma(t))=q^{\mathrm{sep}}(t)$ and rearranging yields
\[
\sigma'(t)
\;=\;
\frac{\big(t-p\,\sigma(t)\big)}{p}\frac{q^{\mathrm{sep}\,\prime}(t)}{q^{\mathrm{sep}}(t)}.
\]
Because $q_\sigma(\sigma(t_1))>0$, the optimum is interior in a right neighborhood of $t_1$, so the FOC applies there; strict monotonicity extends it throughout $(t_1,t_2)$. This coincides with \eqref{eq:ODE-raw}.
\end{proof}

\subsection{Equilibrium-Dominance and its implications}

\paragraph{Equilibrium-Dominance definition.}
For an off-path $s$, a continuation $(\iota',\mu')$ is \emph{admissible} at $s$,  $(\iota',\mu') \in \mathcal{C}(s)$,  if (i) $\iota'$ is DSIP$_{\mu'}$, and (ii) on–path beliefs are Bayes–consistent: for $s\in S_\sigma$, \( \mu(t\mid s)=\frac{f(t)}{\int_{T_\sigma(s)} f(x)\,dx}\). A type-$t$ never wants to deviate and offer $s$, $t \in D_{\sigma, \iota}(s)$ if for any admissible continuation ,
\[
(t-ps)\,q_{\sigma,\iota'}(s)\;<\;(t-p\sigma(t))\,q_{\sigma,\iota'}(\sigma(t))=\pi_{\sigma,\iota}(\sigma(t);t),
\] 
An equilibrium $(\sigma,\iota,\mu)$ \emph{fails equilibrium-dominance} if there exists $s\in S_\sigma^-$ and $t\in D_{\sigma, \iota}(s)$ with $t\in\mathrm{supp}(\mu(s))$.

\paragraph{Optimistic beliefs.}
For some off-path $s$, we refer as \emph{optimistic beliefs at $s$} the belief system where subsidy $s$ is assigned a degenerate “type–1” posterior, while every other off–path $s'\neq s$ is assigned a degenerate “type–0” posterior and never inspected. On–path beliefs remain consistent with $\sigma$, so this constitutes an admissible continuation. Under such beliefs, the inspection index reduces to $r_{\mu'}(s)=u-(c-s)$. Hence any $s>c-u$ is inspected with positive probability, though $q_{\sigma,\iota'}(s)<1$ whenever there exists $s'\in S_\sigma$ with $s'>s$ and $(c-s')/(c-s)<\tau_{\mu'}(s')$.

\begin{proof}[Proof of \cref{lem:EDR-characterization}]
Optimistic beliefs at $s$ are an admissible continuation by construction. Moreover, this belief system places the higher match posterior at $s$, thus yields the highest possible attention that is consistent on path with $\sigma$.
That is, for any $(\iota',\mu') \in \mathcal{C}(s)$, $q_{\sigma, s}(s) \geq q_{\sigma, \iota'}(s)$. Thus if 
\[ (t - ps)q_{\sigma, s}(s) \geq (t - p \sigma (t)) q_{\sigma, \iota}(\sigma(t))  \]
Any type $t'>t$ such that $\sigma(t) = s' = \sigma(t')$ benefits more from such deviation as $\pi$ satisfies increasing differences in $s$ and $q$. Thus, it should hold that $\mu(s) = \sup \{t: \sigma(t) = s'\}$; otherwise the equilibrium fails the refinement. But, if $\mu(s) = \sup \{t: \sigma(t) = s'\}$, then $\mathbb{E}_{t \sim\mu(s)} [t] > \mathbb{E}_{t \sim\mu(s')} [t]$; which implies that $r_\mu(s) > r_\mu (s')$ and hence, $q_{\sigma, \iota}(s) > q_{\sigma, \iota}(s')$. From increasing differences, then if $t$ finds is at least as good in its equilibrium value, any $t' >t$ has a profitable deviation. Hence, $(\sigma, \iota, \mu)$ cannot be an equilibrium. 
\end{proof}

\subsubsection{Consequences of failure}

The following lemmas identify structures that necessarily lead to failure of the refinement. 

Non–monotone indices are incompatible with equilibrium-dominance. In equilibrium, attention must be weakly increasing.

\begin{lemma}\label{lem-app:index-monotone}
Fix an equilibrium $(\sigma,\iota,\mu)$ and let $r_\mu(s)=u-\tfrac{c-s}{\tau_\mu(s)}$.  
If there exist $s',s\in S$ with $s'>s$ but $r_\mu(s')<r_\mu(s)$, then the equilibrium fails equilibrium-dominance.
\end{lemma}

\begin{proof}
It cannot be that both $s',s\in S_\sigma$, since \cref{thm:sorting} already rules out such non–monotonicity on–path. If instead $s'\in S_\sigma$, then the type issuing $s'$ has a profitable deviation, so again the equilibrium fails. The only relevant case is $s'\in S_\sigma^-$.  

Under optimistic beliefs at $s'$, we have $r_\mu(s')>0$.\footnote{We assume generically that $u>c$, which guarantees that even a zero subsidy is inspected under ``type–1” beliefs. In the main text we already impose the slightly weaker condition $u>c-1/p$.} In that case, a type $t\in\mathrm{supp}(\mu(s'))$ deviating from $s'$ to $s$ strictly improves both rank (earlier inspection) and unit margin $t-ps$. Hence $s'$ cannot be a weak best response for $t$, so equilibrium-dominance forbids $\mu(s')$ assigning positive mass to such types. Contradiction.
\end{proof}

Equilibrium-dominance rules out types with negative margins at inspected subsidies.

\begin{lemma}\label{lem-app:no-negative-margin}
Fix an equilibrium $(\sigma,\iota,\mu)$.  
If some $s$ is inspected with positive probability, i.e.\ $q_{\sigma,\iota}(s)>0$, then the equilibrium fails equilibrium-dominance if there exists $t<ps$ with $t\in\mathrm{supp}(\mu(s))$.  
In particular, if $\tau_\mu(s)<ps$ whenever $q(s)>0$, then the equilibrium fails equilibrium-dominance.
\end{lemma}

\begin{proof}
If $t<ps$, then $t-ps<0$. Such a type never weakly prefers any inspected $s$. Hence equilibrium-dominance rules out $\mu(s)$ placing mass on $t$.
\end{proof}

Any interior pooling implies a failure of equilibrium-dominance.

\begin{lemma}\label{lem-app:no-subcap-pooling}
Fix an equilibrium $(\sigma,\iota,\mu)$. If there exists $s \in (0,c)$ such that $\Pr\{t \in T_\sigma(s)\} > 0$, then $(\sigma,\iota,\mu)$ fails equilibrium-dominance.
\end{lemma}

\begin{proof}

Assume, for a contradiction, that $(\sigma,\iota,\mu)$ does not fail equilibrium-dominance.  
By \cref{lem-app:index-monotone}, the index $r_\mu(\cdot)$ and the induced attention $q_\sigma(\cdot)$ are weakly increasing on $[0,c]$.

Let $t^*=\sup T_\sigma(s)$ be the highest type pooled at $s$. Consider a increasing sequence $t^k\uparrow t^*$ with $t^k\in T_\sigma(s)$ for all $k$. By definition, for every $\epsilon >0$, there exists $N(\epsilon)$ such that for all $k \geq N(\epsilon)$, $|t^k - t^*| < \epsilon$. 

From \cref{lem-app:subcap-pool-disc} there exists some $\varepsilon> 0$ such that $(s,s+\varepsilon)\cap S_\sigma=\emptyset$. Then $s' \in (s, s+ \varepsilon)$ is off–path and there exists a fix $q' \in [0,1]$ such that with optimistic beliefs at $s'$, $q' = q_{\sigma, \iota'}(s') > q_{\sigma, \iota'}(s)$.%
\footnote{In particular, $q' = \lim_{t \downarrow t^*}q^\textrm{sep}(t)$ and $q' = 1$ if $\sigma(1) = s$. This follows as in equilibrium, a firm is inspected only when all the firms that rank better fail to match. That is what $q^\textrm{sep}$ captures (see \cref{app:inspection-prob}). } %
Thus, for any $t^k$
\[ (t^k - ps')q' - (t^k - ps)q(s) = (q' - q(s))(t^* - ps) - (q' - q(s))(t^* - t^k) - pq' (s' -s).  \]
Since we are on an equilibrium, for $k > 1$, $t^k > ps$, and thus $\lim_k (t^k - ps) = \pi^* > 0$. Then, fix $\epsilon_1 < \pi^*/2$ and set $s' = s + \varepsilon_2$ with $\varepsilon_2 < (q' - q(s))\pi^*/(2(pq') <0$, if $k > N(\epsilon_1)$, then 
\[ (t^k - ps')q' - (t^k - ps)q(s) = (q' - q(s))\pi^* - (q' - q(s))\epsilon_1 - pq' \epsilon_2 > 0.  \]
That is, type-$t^k$ has a profitable deviation. Contradiction

\end{proof}

We next pin down the inspected portion of $\sigma$ and its lower boundary. Let $\underline{t}=\frac{pc}{1+p}$ be the marginal type that is just worth inspecting and just willing to subsidize.

\begin{lemma}\label{lem-app:lower-boundary}
Fix an equilibrium $(\sigma,\iota,\mu)$.  If the lowest inspected type exceeds $\underline{t}$, ie $t_0 = \inf \{ t \in T: q_\sigma(t)> 0\}>  \underline{t}$, the equilibrium fails equilibrium-dominance. 
\end{lemma}

\begin{proof}
Suppose inspection starts at some $t_0>\underline{t}$. Let $s_0=\sigma(t_0)$ be the lowest on–path subsidy with positive attention. 

First, if $\pi_{\sigma, \iota}(s_0; t_0) = \pi_0 > 0$, then there exists some $t\in(\underline{t},t_0)$ such that $(t - ps_0) > 0$ and its attention is $q(s_0)>0$ by construction. Hence the deviation yields strictly positive profit for $t$, $\pi_{\sigma, \iota}(s_0, t) > 0  = \pi_{\sigma, \iota}(\sigma(t), t)$.

Second, if $\pi_0 = 0$, then type $t_0$ has a profitable deviation to $s \in (c/(1+p), s_0)$. This follows as $t_0 > pc/(1+p)$, thus for such set $(t_0 - ps) > 0$ and under optimistic beliefs at $s$, $q_{\sigma, \iota'}(s) >0$. Thus, $\pi_{\sigma, \iota'}(s; t_0) > 0 = \pi_0$. Since $q_\sigma(s) = 0$, that means that there exists at least another type $t' \in \textrm{supp}(\mu(s))$ such that $t' < ps$. Otherwise, $r_\mu(s) > 0$ and hence inspected. But, $t'$ has a negative margin at $s$ and from \cref{lem-app:no-negative-margin}, this equilibrium fails equilibrium-dominance.
\end{proof}


\subsection{Survival and uniqueness}

\begin{lemma}
The reasonable equilibrium survives the equilibrium–dominance refinement.
\end{lemma}

\begin{proof}
Suppose, for a contradiction, that there exists an off–path message that violates equilibrium–dominance. Any such message must fall into one of two cases:  
(i) $s \in (0,\underline{t}/p)$, or  
(ii) $s \in (\lim_{\varepsilon \uparrow 0}\sigma^*(t_1 + \varepsilon),c)$ when there is a discontinuity at the top.  

Consider case (i). All types $t\leq \underline{t}$ earn nonnegative profit in equilibrium. If any of them were to deviate to an off–path message in $(0,\underline{t}/p)$, their profit would be strictly negative, even under optimistic beliefs. Hence no profitable deviation exists.  

Now consider case (ii). At the reasonable equilibrium, beliefs already assign probability one to the highest possible type at $s=\sigma^*(t_1)$. Thus optimistic beliefs at any $s\in(\sigma^*(t_1),c)$ cannot make the deviation more attractive. A type that raises its subsidy above $\sigma^*(t_1)$ obtains the same attention as $t_1$ but must pay strictly more, which strictly reduces profit. Lower–type firms do not gain either: although attention may rise, the higher subsidy cost dominates, again reducing profit.  

In both cases, no type finds it profitable to deviate to an off–path subsidy. Therefore, the reasonable equilibrium survives the equilibrium–dominance refinement.
\end{proof}

Therefore, we are ready to establish the proof for the main result of the section.

\begin{proof}[Proof of \cref{thm:uniqueness}]
By \cref{lem-app:no-subcap-pooling}, the policy in any refined equilibrium is step–increasing–step: no inspection below some cutoff, then a strictly increasing revealing branch, then a single pooling block at $c$. By \cref{lem-app:lower-boundary} the lower cutoff must be $\underline t$. By \cref{lem-app:revealing-ODE}, the revealing branch is characterized by the same truth–telling ODE and the same boundary condition at $\underline t$, hence it coincides with $\sigma^*$ on $[\underline t,\bar t)$. By \cref{lem-app:no-negative-margin} pooling can occur only at $c$, and by \cref{lem:existence-cutoff} the pooling cutoff $\bar t$ is uniquely pinned by boundary indifference. Thus the outcome (cutoffs, revealing schedule, and top pooling) is unique and matches the one induced by the reasonable equilibrium $(\sigma^*,q^*, \mu^*)$.
\end{proof}

\section{Omitted proofs for \cref{sec:welfare}} \label{appendix:welfare-proofs}

\subsection{Some useful identities}

First, we show how to derive the simplifications presented for the consumer surplus and the producer surplus.

\begin{lemma}\label{lem-app:useful-identities}
    The producer surplus can be simplified as follows:
    \begin{equation}
        \frac{PS_n(c,p)}{n} = \int_{\underline{t}}^1 (1 - F(t)) q^*(t) \, dt = \mathbb{E}_t\left[\frac{1 - F(t)}{f(t)}q^*(t)\right],
    \end{equation}
    while the expected firms transfers are given by 
    \begin{equation}
    \begin{split}
        \phi_n^F(p,c) & = \int_{\underline{t}}^{\bar t } \left(t - \frac{1 - F(t)}{f(t)}\right)q^\textrm{sep}(t)f(t) \, dt + [1 - F(\bar t)]\big( \bar tq^\textrm{pool}(\bar t)\big)
        \\ & = \mathbb{E}_t\left[ \psi(t) q^\textrm{sep}(t; \underline{t})\; \Big|\; t \leq \bar t \right] F(\bar t) + \big( \bar tq^\textrm{pool}(\bar t)\big)[1 - F(\bar t)]
    \end{split}
    \end{equation}
\end{lemma}

\begin{proof}
Recall that:
\[ \begin{split}
    \phi_n^F(p,c) & = \mathbb{E}_t[p\sigma^*(t)q^*(t)] 
    \\ & = \int_{\underline{t}(p,c)}^{\bar t (p,c)} \left(tq^\textrm{sep}(t) - \int_{\underline{t}(p,c)}^tq^\textrm{sep}(x)\, dx\right)f(t) \, dt + pcq^\textrm{pool}(\bar t(p,c))[1 - F(\bar t (p,c))]
\end{split}\]
This expression can be further simplified. First note that
\[ \begin{split}
    \int_{\underline{t}(p)}^{\bar t (p)} \Big(tq^\textrm{sep}(t) & - \int_{\underline{t}(p)}^tq^\textrm{sep}(x)\, dx\Big) f(t) \, dt 
    \\& = \int_{\underline{t}(p)}^{\bar t (p)} \left(t - \frac{F(\bar t (p)) - F(t)}{f(t)}\right)q^\textrm{sep}(t)f(t) \, dt
    \\ & = \int_{\underline{t}(p)}^{\bar t (p)} \left(t - \frac{1 - F(t)}{f(t)}\right)q^\textrm{sep}(t)f(t) \, dt + [1 - F(\bar t(p))]\int_{\underline{t}(p)}^{\overline{t}(p)}q^\textrm{sep}(t)\;dt
    \\ & = \int_{\underline{t}(p)}^{\bar t (p)} \left(t - \frac{1 - F(t)}{f(t)}\right)q^\textrm{sep}(t)f(t) \, dt + [1 - F(\bar t(p))]\big( \bar t(p) - pc \big)q^\textrm{pool}(\bar t(p)),
\end{split} \]
where the first equality uses integration by parts, the second follows from algebraic manipulation and the third from the \ref{eq:boundary} condition. Therefore, by pasting this back to the original expression it yields
\begin{equation*}
\begin{split}
    \phi_n^F(p,c) & = \int_{\underline{t}}^{\bar t } \left(t - \frac{1 - F(t)}{f(t)}\right)q^\textrm{sep}(t)f(t) \, dt + [1 - F(\bar t)]\big( \bar tq^\textrm{pool}(\bar t)\big)
    \\ & = \mathbb{E}_t\left[ \psi(t) q^\textrm{sep}(t; \underline{t})\; \Big|\; t \leq \bar t \right] F(\bar t) + \big( \bar tq^\textrm{pool}(\bar t)\big)[1 - F(\bar t)]
\end{split}
\end{equation*}

Similarly, we can simplify the expression for the producer surplus. Recall that 
\[ \begin{split}
    PS(p,c) & = \mathbb{E}_t[(t - p\sigma^*(t))q^*(t)]
    \\ & = \int_{\underline{t}}^{\bar t} (t - p \sigma^\textrm{sep}(t))q^\textrm{sep}(t) f(t) \,dt + \int_{\bar t}^1 (t - pc)q^\textrm{pool}(\bar t)f(t)\, dt
    \\ & = \int_{\underline{t}}^{\bar t} \left(\int_{\underline{t}}^t q^\textrm{sep}(x) dx \right)f(t) \,dt + q^\textrm{pool}(\bar t)\left(\int_{\bar t}^1 t f(t)\, dt -  pc[ 1 - F(\bar t)]\right)
    \\ & = \int_{\underline{t}}^{\bar t} \big[F(\bar t) - F(t)\big]q^\textrm{sep}(t) \,dt + q^\textrm{pool}(\bar t)\left(\int_{\bar t}^1 t f(t)\, dt -  pc[ 1 - F(\bar t)]\right)
\end{split} \]
where the first equality is definition, the second is achieved by substitution, the third uses the definition and algebraic manipulation and the fourth integration by parts. Re-arranging we obtain
\[ PS_n(p,c) = \int_{\underline{t}}^{\bar t} \big[1 - F(t)\big]q^\textrm{sep}(t) \,dt + q^\textrm{pool}(\bar t)\int_{\bar t}^1 t f(t)\, dt -  [ 1 - F(\bar t)]\left(\int_{\underline{t}}^{\bar t} q^\textrm{sep}(t)\,dt+ q^\textrm{pool}(\bar t)pc\right) \]
To conclude, note that from the \ref{eq:boundary}
\[ \int_{\underline{t}}^{\bar t} q^\textrm{sep}(t)\,dt+ q^\textrm{pool}(\bar t)pc = (\bar t - pc)q^\textrm{pool}(\bar t) + q^\textrm{pool}(\bar t)pc = \bar t q^\textrm{pool}(\bar t)\]
Thus, 
\[ \begin{split}
    PS_n(p,c) & = \int_{\underline{t}}^{\bar t} \big[1 - F(t)\big]q^\textrm{sep}(t) \,dt + q^\textrm{pool}(\bar t)\left[ \int_{\bar t}^1 t f(t)\, dt -  [ 1 - F(\bar t)]\bar t \right]
    \\ &= \int_{\underline{t}}^{\bar t} \big[1 - F(t)\big]q^\textrm{sep}(t) \,dt + q^\textrm{pool}(\bar t)\left[1 - \bar t F(\bar t) -  \int_{\bar t}^1 F(t)\, dt -  [ 1 - F(\bar t)]\bar t \right]
    \\ & = \int_{\underline{t}}^{\bar t} \big[1 - F(t)\big]q^\textrm{sep}(t) \,dt + \int_{\bar t}^1[1 - F(t)]q^\textrm{pool}(\bar t)\, dt 
\end{split}\]
Therefore,
\begin{equation*}
    PS_n(c,p) = \int_{\underline{t}}^1 (1 - F(t)) q^*(t) \, dt = \mathbb{E}_t\left[\frac{1 - F(t)}{f(t)}q^*(t)\right]
\end{equation*}
\end{proof}

\begin{lemma}\label{lem-app:match-prob}
    Fix an equilibrium $(\sigma, \iota, \mu)$ and let $t_0 = \inf \{t\in T: q_{\sigma, \iota}(\sigma(t)) > 0\}$. The expected match probability is
    \[ m(\sigma, \iota) = 1 - \left(1 - \int_{t_0}^1 tf(t)\;dt\right)^n. \]
\end{lemma}
\begin{proof}
    A match occurs if, and only if, one inspected firm succeeds. Hence, the probability of no match equals the probability that all eligible firms fail.

    Let $K\sim \mathrm{Bin}\!\left(n,\,1-F(t_0)\right)$ be the number of firms that are positively inspected and  $\tau \equiv \mathbb{E}\!\left[t \,\middle|\, t\ge t_0\right]$ the mean success probability conditional on inspection. Conditional on $K=k$, the $k$ eligible firms are i.i.d., so the probability they all fail is $(1-\tau)^k$. Averaging over $K$ yields
    \[ \begin{split}
        \Pr\{\text{no match}\}
         & \;=\;
        \sum_{k=0}^{n} \binom{n}{k}\!\left(1-F(t_0)\right)^{k} \!\left(F(t_0)\right)^{n-k} (1-\tau)^{k}
        \\ & \;=\;
        \Bigl(F(t_0) + (1-F(t_0))(1-\tau)\Bigr)^{n} \\
        & \;=\; \big(1 - \mathbb{E}[t \mid t \geq t_0]\Pr\{t \geq t_0\}\big)^n 
    \end{split}
    \]
    Therefore,
    \begin{equation}\label{eq:match-simple}
    m(\sigma, \iota) \equiv \Pr\{\text{match}\}
    \;=\;
    1 - \Bigl(1 - \int_{t_0}^1 t \,dF(t)\Bigr)^{n}.
    \end{equation}
\end{proof}

\subsection{Omitted proofs for \cref{sec:welfare}}

\begin{proof}[Proof of \cref{lem:SIS-properties}]
\emph{(i) Lower cutoff.}  
The expression for the lower cutoff $\underline{t}(p,c)=pc/(1+p)$ follows directly from the consumer’s participation condition. It is strictly increasing in $c$ and $p$, and by inspection does not depend on $n$.

\emph{(ii) Effect of market size.}  
Recall that
\[
q_n^{\mathrm{sep}}(t;t_0) \;=\; \Big(1-\int_t^1 x f(x)\,dx\Big)^{n-1}\,\1\{t\ge t_0\}.
\]
Since the term in parentheses is less than one, $q_n^{\mathrm{sep}}(t;t_0)$ is strictly decreasing in $n$ for each $t\in[t_0,1)$. Moreover, by construction $q_n^{\mathrm{sep}}(t;\cdot)$ is decreasing in $t$.
Moreover, note then that 
\[ \frac{\partial}{\partial c} q_n^{\mathrm{sep}}(t;\underline{t}) = \frac{\partial}{\partial t_0} q_n^{\mathrm{sep}}(t;\underline{t}) \frac{\partial}{\partial c} \underline{t}   \leq 0\]
Similarly,
\[ \frac{\partial}{\partial p} q_n^{\mathrm{sep}}(t;\underline{t}) = \frac{\partial}{\partial t_0} q_n^{\mathrm{sep}}(t;\underline{t}) \frac{\partial}{\partial p} \underline{t}   \leq 0.\]

\emph{(iii) Monotonicity of policies.}  
The separating subsidy policy is given by
\[
\sigma_n^{\mathrm{sep}}(t;t_0) \;=\; \frac{1}{p}\left( t - \frac{\int_{t_0}^t q_n^{\mathrm{sep}}(x)\,dx}{q_n^{\mathrm{sep}}(t)} \right).
\]
It is immediate that if $t_0$ increases, then for any $t\ge t_0$ the numerator $\int_{t_0}^t q_n^{\mathrm{sep}}(x)\,dx$ falls, and hence $\sigma_n^{\mathrm{sep}}(t;t_0)$ rises. This proves monotonicity in the cutoff.  

To compare across market sizes, note that
\[
\sigma_{n+1}^{\mathrm{sep}}(t;t_0)-\sigma_n^{\mathrm{sep}}(t;t_0)
= \frac{1}{p}\left( \frac{\int_{t_0}^t q_n^{\mathrm{sep}}(x)\,dx}{q_n^{\mathrm{sep}}(t)} - \frac{\int_{t_0}^t q_{n+1}^{\mathrm{sep}}(x)\,dx}{q_{n+1}^{\mathrm{sep}}(t)} \right).
\]
Since
\[
q_{n+1}^{\mathrm{sep}}(t) = q_n^{\mathrm{sep}}(t)\left(1-\int_t^1 x f(x)\,dx\right),
\]
and $q_n^{\mathrm{sep}}(\cdot)$ is increasing in $t$, the difference above is strictly positive. Hence $\sigma_{n+1}^{\mathrm{sep}}(t;t_0)>\sigma_n^{\mathrm{sep}}(t;t_0)$ for all $t>t_0$.

To conclude, note that for $t > \underline{t}$,
\[ \frac{\partial}{\partial c} \sigma_n^{\mathrm{sep}}(t;\underline{t}) = \frac{\partial}{\partial t_0} \sigma_n^{\mathrm{sep}}(t;\underline{t}) \frac{\partial}{\partial c} \underline{t}   \geq 0\]
Moreover
\[ \begin{split}
    \frac{\partial}{\partial p} \sigma^\textrm{sep}(t; \underline{t})  & =  \frac{\partial}{\partial p} \left( \frac{t}{p} - \frac{\int_{\frac{cp}{1+p}}^t q^\textrm{sep}(x)\; dx}{p q^\textrm{sep}(t)}\right) 
    \\ & = - \frac{t}{p^2} - \left( \frac{- q^\textrm{sep}(\underline{t}) (c/(1+p)^2) - q^\textrm{sep}(t)^2 ( t - p\sigma^\textrm{sep}(t)) }{p^2 q^\textrm{sep}(t)^2}\right)
    \\ & = -\frac{1}{p^2} \left(   p \sigma^\textrm{sep}(t; \underline{t}) - \frac{q^\textrm{sep}(\underline{t})}{q^\textrm{sep}(t)^2}\frac{c}{(1+ p)^2} \right)
\end{split}  \]
Moreover, 
\[  p \sigma^\textrm{sep}(t; \underline{t}) > \frac{q^\textrm{sep}(\underline{t})}{q^\textrm{sep}(t)^2}\frac{c}{(1+ p)^2}   \iff p \sigma^\textrm{sep}(t; \underline{t})q^\textrm{sep}(t) > \frac{c}{(1 +p)^2} = \frac{\underline{t}}{p (1 +p)}. \]
Therefore, if $p(1+p) \geq 1$, this last equality is true and thus $\frac{\partial}{\partial p} \sigma^\textrm{sep}(t; \underline{t})< 0$
Otherwise, as $p(1+p) < 1$, then $- \underline{t} > - \tfrac{c}{(1 +p)^2}$, and we have that 
\[  p \sigma^\textrm{sep}(t; \underline{t}) > \frac{q^\textrm{sep}(\underline{t})}{q^\textrm{sep}(t)^2}\underline{t} \implies   p \sigma^\textrm{sep}(t; \underline{t}) > \frac{q^\textrm{sep}(\underline{t})}{q^\textrm{sep}(t)^2}\frac{c}{(1+ p)^2}  \]
We thus establish the result. 

\emph{(iv) Cutoff monotonicity.}  
As $n$ increases, subsidies rise and the subsidy cap $c$ is reached at a lower type. Similarly, an increase in $p$ or $c$ shifts the incentive constraints so that the upper cutoff $\bar t$ is attained more quickly. Thus $\bar t$ is strictly decreasing in both $c$ and $p$.

This completes the proof.
\end{proof}

\begin{proof}[Proof of \cref{thm:comparative-statics-SIS}]

Let $\bar q_n(p,c):=Q_n(p,c)/n$ be the per–firm inspection probability.

\medskip\noindent\textit{(i) Search intensity.}
Fix $(p,c)$ and compare $n$ and $n+1$. Since $\bar q_n$ is the per–firm inspection probability, $\bar q_{n+1}\le \bar q_n$ by \cref{lem:SIS-properties}(ii). However, total inspections
\[
Q_n \;=\; n\,\bar q_n \qquad\text{and}\qquad Q_{n+1} \;=\; (n+1)\,\bar q_{n+1}
\]
satisfy $Q_{n+1}\ge Q_n$.

Therefore $Q_n$ is (weakly) increasing in $n$. Next, holding $n$ fixed and raising $p$ or $c$ increases $\underline t$ and weakly expands the top pool, so $q^*(t)$ weakly falls for every $t$; hence $Q_n$ is decreasing in $p$ and $c$.

\medskip\noindent\textit{(ii) Match probability.}
As the price of subsidy or search cost increases, the number of participating types in equilibrium lowers (see \cref{lem:SIS-properties}). Thus, the match probability decreases. However, as the number of competing firms increases, the match probability increases as the failure event is lower. That means, match probability increases in $n$.

\medskip\noindent\textit{(iii) Consumer surplus.}
The consumer surplus is decreasing in $p$. This follows as the match probability decreases as $p$ increases, and also firms transfers. A similar argument shows that as compeititon increases, consumer surplus also increases. 

\medskip\noindent\textit{(iv) Producer surplus.}
From \cref{lem-app:useful-identities}, it follows that 
\[ PS_n(c,p) = \mathbb{E}_t\left[\frac{1 - F(t)}{f(t)}q_n^*(t; p,c)\right]\]
where $q_n^*(t; p,c)$ is the attention function (equal to the separating one for $t \in [t_0, t_1)$ and equal to the pooling one in $t \in (t,1]$).

Write \(w(t):=1-F(t)\) and, for any attention schedule \(q\), its cumulative \(Q_n(x):=\int_0^x q_n^*(s; p,c )\,ds\).
Integration by parts yields
\[
\int_0^1 w(t)q(t)\,dt = w(1)Q(1) - \int_0^1 w'(t)\,Q(t)\,dt.
\]
Since \(w'(t)=-f(t)\le 0\) and \(w(1)=1-F(1)=0\) (support \([0,1]\)), this reduces to
\[
PS_n(c,p) \;=\; \int_0^1 w(t)q_n^*(t;p,c)\,dt
            \;=\; -\int_0^1 w'(t)\,Q_{n}(t; p,c)\,dt
            \;=\; \int_0^1 f(t)\,Q_{n}(t; p,c )\,dt.
\]

(i) Fix \(n,c\). If \(p'>p\) then following \cref{lem:SIS-properties}, \(Q_{n}(t; p',c)\le Q_p(t; p,c)\) for all \(t\). Because \(f(t)\ge 0\) for all \(t\),
\[
PS_n(c,p') \;=\; \int_0^1 f(t)\,Q_{n}(t; p',c)\,dt
              \le \int_0^1 f(t)\,Q_n(t;p,c)\,dt
              \;=\; PS_n(c,p).
\]

(ii) The same argument applies to \(c\): if \(c'>c\) then \(Q_{n}(t; p,c')\le Q_n(t; p,c)\) for all \(t\), hence \(PS_n(c',p)\le PS_n(c,p)\).

(iii) Fix \(p,c\). Once again, \(Q_{n+1}(t; p,c)\ge Q_n(t; p,c)\) for all \(t\), then
\[
PS_{n+1}(c,p)=\int_0^1 f(t)\,Q_{n+1}(t)\,dt
            \ge \int_0^1 f(t)\,Q_n(t)\,dt
            = PS_n(c,p),
\]
so aggregate producer surplus is nondecreasing in \(n\).

\end{proof}

\section{Omitted proofs for \cref{sec:agentic}} \label{appendix:agentic-proofs}

\begin{proof}[Proof of \cref{thm:excess-search}.]
(i) If pooling were inactive at the maximizer (so \(\bar t=1\)), a small decrease in \(p\) would create a positive-measure pool and increase the pool term in \eqref{eq:rev-decomp} at first order, while affecting the revelation term only at second order through the cutoff; revenue would rise, a contradiction.

(ii) Let \(t_\psi\) solve \(\psi(t_\psi)=0\) (exists and is unique under regularity). If \(\bar t(p^\star)\ge t_\psi\), then reducing \(p\) slightly lowers \(\bar t\) and expands the pool into the region with \(\psi<0\). Using boundary indifference and the envelope logic behind \eqref{eq:rev-decomp}, the marginal revenue gain from the pool dominates the marginal loss on the revelation term at that point, so revenue rises. Hence at an interior maximum we must have \(\bar t(p^\star)<t_\psi\), and some \(\psi<0\) types are inspected (in the pool or just below it).

\end{proof}

\newpage






\end{document}